%% file: main.tex
\documentclass[conference,compsoc]{IEEEtran}
\ifCLASSOPTIONcompsoc
  \usepackage[nocompress]{cite}
\else
  \usepackage{cite}
\fi
\ifCLASSINFOpdf
\else
\fi

\usepackage{amsmath}
\usepackage{amssymb}

\usepackage{booktabs,tabularx,caption}
\usepackage{listings}
\usepackage{listings-rust}%
\usepackage{xcolor}
\usepackage{algorithmicx}%
\usepackage{algorithm}
\usepackage[noend]{algpseudocode}%
\usepackage{pifont}
\usepackage{colortbl} % For coloring table cells
\usepackage{tikz}[node]%
\usepackage{wrapfig}
\usepackage{tcolorbox}%
\usepackage{colortbl}
\usepackage[table]{xcolor}
\usepackage[nomessages]{fp}% http://ctan.org/pkg/fp
\usepackage{xfp}%
\usepackage{placeins}%
\usepackage{xspace}
\usepackage{adjustbox}
\usepackage{ulem}%
\usepackage{amsthm}
\usepackage{hyperref}
\usepackage{ulem}

\usepackage{enumitem}

\newtheorem{definition}{Definition}
\newtheorem{theorem}{Theorem}

\newtheorem{lemma}[theorem]{Lemma}

\definecolor{myBlue}{RGB}{25, 118, 210}
\definecolor{onyx}{rgb}{0.06, 0.06, 0.07}
\definecolor{charcoal}{rgb}{0.1, 0.1, 0.1}

\newcommand{\cellcolortale}[1]{\cellcolor{teal!\fpeval{#1*0.6}}}
\newcommand{\cellcolorred}[1]{\cellcolor{red!\fpeval{#1*0.6}}}

\lstdefinelanguage{Circom}{
    keywords=[1]{signal, component, input, output, var, function, return, include},
    keywords=[2]{template, main, true, false},
    keywordstyle=[1]\color{myBlue}\bfseries,
    keywordstyle=[2]\color{purple}\bfseries,
    identifierstyle=\color{black},
    commentstyle=\color{gray}\ttfamily,
    stringstyle=\color{orange}\ttfamily,
    sensitive=true,
    morecomment=[l]{//},
    morecomment=[s]{/*}{*/},
    morestring=[b]",
}
\lstdefinestyle{circomstyle}{
    backgroundcolor=\color{gray!0},
    rulecolor=\color{gray!0},
}

\newlength{\maxlen}

\newcommand{\hquad}{\hspace{0.5em}} 

\newcommand{\cmark}{\textcolor{green!70!black}{\ding{51}}} % Checkmark
\newcommand{\xmark}{\textcolor{red}{\ding{55}}} % Cross
\newcommand{\dangerous}{\textcolor{red!80!black}{\textbf{Yes}}}
\newcommand{\safe}{\textcolor{green!80!black}{\textbf{No}}}

\newcommand*\partialcircle[1][1ex]{%
  \begin{tikzpicture}
  \draw[fill=green!70] (0,0)-- (90:#1) arc (90:270:#1) -- cycle ;
  \draw (0,0) circle (#1);
  \end{tikzpicture}}

\newcommand{\rqbox}[2]{
    \begin{tcolorbox}[colback=gray!10,colframe=black,boxrule=0.5pt]
        #2
    \end{tcolorbox}
}

\newcommand{\sys}{\textsc{zkFuzz}\xspace}
\newcommand{\sysb}{\textsc{zkFuzz}\xspace}
\newcommand{\sysp}{\textsc{zkFuzz++}\xspace}
\newcommand{\model}{\textsc{TCCT}\xspace}
\newcommand{\nallcircuits}{452\xspace}
\newcommand{\nallbugs}{88\xspace}
\newcommand{\ndetectedbugs}{85\xspace}
\newcommand{\nnewbugs}{59\xspace}
\newcommand{\nconfirmed}{39\xspace}
\newcommand{\nconfirmedzkregex}{11\xspace}
\newcommand{\nfixed}{14\xspace}

\newif\ifpreprint

\newif\ifcomment
\commentfalse

\begin{document}
%
% paper title
% Titles are generally capitalized except for words such as a, an, and, as,
% at, but, by, for, in, nor, of, on, or, the, to and up, which are usually
% not capitalized unless they are the first or last word of the title.
% Linebreaks \\ can be used within to get better formatting as desired.
% Do not put math or special symbols in the title.
\title{\sys: Foundation and Framework \\ for  Effective Fuzzing of Zero-Knowledge Circuits}

% author names and affiliations
% use a multiple column layout for up to three different
% affiliations

\author{\IEEEauthorblockN{Hideaki Takahashi, Jihwan Kim, Suman Jana, Junfeng Yang}
\IEEEauthorblockA{Columbia University\\
New York, New York 10027\\
Email: ht2673@columbia.edu, jk4908@columbia.edu, suman@cs.columbia.edu, junfeng@cs.columbia.edu}
%\IEEEauthorblockN{Homer Simpson}
%\IEEEauthorblockA{Twentieth Century Fox\\
%Springfield, USA\\
%Email: homer@thesimpsons.com}
%\and
%\IEEEauthorblockN{James Kirk\\ and Montgomery Scott}
%\IEEEauthorblockA{Starfleet Academy\\
%San Francisco, California 96678-2391\\
%Telephone: (800) 555--1212\\
%Fax: (888) 555--1212}
}

% conference papers do not typically use \thanks and this command
% is locked out in conference mode. If really needed, such as for
% the acknowledgment of grants, issue a \IEEEoverridecommandlockouts
% after \documentclass

% for over three affiliations, or if they all won't fit within the width
% of the page (and note that there is less available width in this regard for
% compsoc conferences compared to traditional conferences), use this
% alternative format:
% 
%\author{\IEEEauthorblockN{Michael Shell\IEEEauthorrefmark{1},
%Homer Simpson\IEEEauthorrefmark{2},
%James Kirk\IEEEauthorrefmark{3}, 
%Montgomery Scott\IEEEauthorrefmark{3} and
%Eldon Tyrell\IEEEauthorrefmark{4}}
%\IEEEauthorblockA{\IEEEauthorrefmark{1}School of Electrical and Computer Engineering\\
%Georgia Institute of Technology,
%Atlanta, Georgia 30332--0250\\ Email: see http://www.michaelshell.org/contact.html}
%\IEEEauthorblockA{\IEEEauthorrefmark{2}Twentieth Century Fox, Springfield, USA\\
%Email: homer@thesimpsons.com}
%\IEEEauthorblockA{\IEEEauthorrefmark{3}Starfleet Academy, San Francisco, California 96678-2391\\
%Telephone: (800) 555--1212, Fax: (888) 555--1212}
%\IEEEauthorblockA{\IEEEauthorrefmark{4}Tyrell Inc., 123 Replicant Street, Los Angeles, California 90210--4321}}

% use for special paper notices
%\IEEEspecialpapernotice{(Invited Paper)}

% make the title area
\maketitle
\renewcommand{\thefootnote}{\fnsymbol{footnote}}
\footnotetext[1]{This paper has been accepted at IEEE Symposium on Security and Privacy (S\&P) 2026.}
\renewcommand{\thefootnote}{\arabic{footnote}}

% As a general rule, do not put math, special symbols or citations
% in the abstract
\begin{abstract}

Zero-knowledge (ZK) circuits enable privacy-preserving computations and are central to many cryptographic protocols. Systems like Circom simplify ZK development by combining witness computation and circuit constraints in one program. However, even small errors can compromise security of ZK programs ---\ under-constrained circuits may accept invalid witnesses, while over-constrained ones may reject valid ones. Static analyzers are often imprecise with high false positives, and formal tools struggle with real-world circuit scale. Additionally, existing tools overlook several critical behaviors, such as intermediate computations and program aborts, and thus miss many vulnerabilities.

{Our theoretical contribution is the Trace-Constraint Consistency Test (\model), a foundational, language-independent formulation of ZK circuit bugs. \model provides a unified semantics that subsumes prior definitions and captures both under- and over-constrained vulnerabilities, exposing the full space of ZK bugs that elude prior tools.}

    Our systems contribution is \sys, a novel program mutation-based fuzzing framework for detecting \model violations. \sys systematically mutates the computational logic of Zk programs guided by a novel fitness function, and injects carefully crafted inputs using tailored heuristics to expose bugs. We evaluated \sys on \nallcircuits real-world ZK circuits written in Circom, a leading programming system for ZK development. \sys successfully identified \ndetectedbugs bugs, including \nnewbugs zero-days—\nconfirmed of which were confirmed by developers and \nfixed fixed, including bugs undetectable by prior works due to their fundamentally limited formulations, earning thousands of bug bounties. Our preliminary research on Noir, another emerging DSL for ZK circuit, also demonstrates the feasibility of \sys to support multiple DSLs.%\footnote{This paper has been accepted to IEEE Symposium on Security and Privacy (IEEE S\&P) 2026.} %To demonstrate \sys's language-independence, we built a preliminary prototype for Noir—a popular ZK DSL unsupported by existing tools—and showed it can detect previously reported bugs.

\end{abstract}

% no keywords

% For peer review papers, you can put extra information on the cover
% page as needed:
% \ifCLASSOPTIONpeerreview
% \begin{center} \bfseries EDICS Category: 3-BBND \end{center}
% \fi
%
% For peerreview papers, this IEEEtran command inserts a page break and
% creates the second title. It will be ignored for other modes.
\IEEEpeerreviewmaketitle

\input{subfiles/introduction}

\input{subfiles/background}
\input{subfiles/bug_definition}

\input{subfiles/design}

\input{subfiles/implementation}

\input{subfiles/evaluation}

\input{subfiles/related_work}

\input{subfiles/limitation_and_futurework}

\input{subfiles/conclusion}

\section*{Acknowledgments}

We sincerely thank the anonymous reviewers and our shepherd for their insightful and constructive feedback. This work was partially supported by the Columbia Center for Digital Finance and Technologies and the Funai Overseas Scholarship Foundation.

\bibliographystyle{IEEEtran}
\bibliography{ref}

\input{subfiles/appendix}

% that's all folks
\end{document}

%% file: subfiles/introduction.tex
\section{Introduction}

Zero-knowledge (ZK) circuits have emerged as a foundational technology for privacy-preserving computation with succinct proofs and efficient verification~\cite{lavin2024survey, morais2019survey, hasan2019overview}. Their adoption spans a variety of critical sectors, including anonymous cryptocurrencies~\cite{ruj2024zero}, confidential smart contracts~\cite{yang2020zero, steffen2022zeestar}, secure COVID-19 contact tracing protocols~\cite{liu2020privacy}, privacy-preserving authentication~\cite{baldimtsi2024zklogin}, and verifiable machine learning systems~\cite{xing2023zero, zhang2024research}. Industry forecasts predict that the ZK market will reach \$10B by 2030, with Web3 applications alone expected to execute nearly 90 billion proofs~\cite{ProtocolLabsZK}.

%{\bf ZK DSL \& compilers.} 
Manually constructing ZK circuits is complex and error-prone, so developers often use programming systems like Circom~\cite{belles2022circom}, which allow them to express both the \textit{computation} for deriving the secret witness from the input signals and the \textit{constraints} of the finite-field arithmetic circuit (used to generate and verify proofs later) within one ZK program written in a high-level, domain-specific language (DSL). The DSL compiler then translates the program into (1) an executable (e.g., in WASM) for witness generation and (2) a set of constraints that define the circuit and support proof verification.

Even with high-level programming systems, developing correct circuits remains notoriously difficult due to two core challenges. First, the computation and the constraints must align exactly, yet they operate in different domains: general-purpose logic vs. polynomial constraints optimized for efficiency. Developers must often manually translate logic into polynomial form, and any mismatch results in incorrect or unverifiable proofs~\cite{wen2024practical}. Second, the non-intuitive behavior of modular arithmetic in finite fields often confuses unfamiliar developers. 

Unsurprisingly, these challenges lead to vulnerabilities: \textit{under-constrained} circuits that allow false proofs, and \textit{over-constrained} circuits that reject valid witnesses. Both can result in serious security breaches~\cite{chaliasos2024sok,liang2025sok}. For instance, an under-constrained bug in zkSync—a widely used zk-Rollup—allowed a malicious prover to extract \$1.9 billion worth of funds at the time of disclosure~\cite{chainlight_zksync_era_write_query_poc,tang2024zero}. {While the risk of under-constrained circuits is widely acknowledged~\cite{pailoor2023automated,liang2025sok,chaliasos2024sok}, over-constrained circuits can also cause severe denial-of-service vulnerabilities for honest provers.} {For example, one over-constrained bug was recently discovered in RISC-Zero, the most widely used zkVM~\cite{veridise2025risc0}, which silently rejected valid execution traces, effectively compromising the integrity of the virtual machine.}

Existing methods for detecting vulnerabilities in ZK programs are limited. Static analysis tools rely on patterns and miss deep semantic issues, leading to high false positives that erode developer trust~\cite{wen2024practical,circomspect,coverity:cacm}. Formal methods offer better precision but don’t scale to large circuits due to SMT solver bottlenecks~\cite{isabel2024scalable,liu2024certifying,pailoor2023automated,chen2024ac4,jiang2025conscs}. Dynamic approaches focus on compilation or verification systems and don’t generalize to individual ZK programs~\cite{hochrainer2024fuzzing,leeb2024metamorphic,xiao2025mtzk,fan2024snarkprobe}.

Furthermore, existing formal definitions of ZK vulnerabilities are incomplete, leaving many bugs undetected. In particular, prior work focused only on constraints, ignoring program semantics, and defined bugs solely as non-determinism in constraints—situations where multiple outputs for the same input all satisfy the constraints~\cite{pailoor2023automated, chen2024ac4, jiang2025conscs}. This omission misses two major classes: over-constrained bugs and under-constrained bugs due to abnormal termination, where inputs can crash witness generation while still satisfying circuit constraints. Together, these overlooked categories account for nearly 50\% of real-world ZK bugs in our evaluation.

Our first contribution in this paper is a new theoretical formulation, the Trace-Constraint Consistency Test (\model), that models the vulnerabilities in ZK programs as discrepancies between $(1)$ all possible execution traces that a computation may produce and $(2)$ the set of input, intermediate, and output values permitted by the circuit constraints. {\model captures both under-constrained, or \textit{soundness}, vulnerabilities—when an invalid trace is accepted by the circuit constraints—and over-constrained, or \textit{completeness}, vulnerabilities—when a valid trace is incorrectly rejected. TCCT provides two major advantages over prior definitions. First, it offers a \textit{unified semantics}: TCCT is the first framework to capture both under- and over-constrained vulnerabilities within the same formal model, correctly accounting for intermediate computations and handling abnormal termination in witness generation. This unification subsumes prior definitions as special cases and exposes the full space of ZK vulnerabilities. Second, TCCT is \textit{language-agnostic}: it operates at the level of trace/constraint properties rather than DSL syntax. While our implementation, zkFuzz, currently supports Circom and has very preliminary support for Noir, many DSLs already compile to standardized formats (e.g., R1CS), and by similarly defining a universal format for execution traces, zkFuzz can readily extend across diverse ZK languages.}

Our second contribution is the design and implementation of \sys, a novel dynamic analysis framework that leverages fuzzing with program mutations to detect \model violations in ZK programs. To uncover over-constrained vulnerabilities, \sys systematically fuzzes inputs to the witness computation in search of valid execution traces incorrectly rejected by the circuit constraints. To uncover under-constrained vulnerabilities, it mutates the witness computation and searches for inputs that produce different outputs accepted by the circuit constraints. Since these outputs are produced by a mutated computation that semantically differs from the original, they are, by construction, invalid and should not be accepted. \sys is fully automated, and for every bug detected, it provides a concrete counterexample, greatly simplifying diagnosis.

%more bugs
A key challenge for \sys is the vast search space. Although ZK program inputs are finite, they belong to large prime fields, and behavior in small fields does not generalize to larger fields (e.g., \texttt{iszero(2 + 3)} is true in $\mathbb{F}_5$ but false in $\mathbb{F}_7$). Blindly searching these large fields with many input signals is computationally expensive. The program mutation space is even larger—practically infinite—since arbitrary programs may produce outputs accepted by buggy circuits. Traditional heuristics like coverage-guided fuzzing are ineffective, as \sys seeks specific traces that expose vulnerabilities, not broad code coverage.

To search this vast space efficiently, \sys adopts a joint input and program mutation-based evolutionary fuzzing algorithm. Given a witness computation program $P$ and circuit constraints $C$, \sys generates multiple mutants of $P$ and heuristically samples multiple input values likely to lead to bugs. Unlike traditional fuzzing, which typically mutates only the inputs, \sys mutates both the program and inputs. It executes all program mutants on all sampled inputs and checks the resulting traces against $C$ for \model violations.

In each iteration, \sys scores program mutants using a novel min-sum fitness function that estimates the likelihood of triggering a violation. For each input, it sums constraint errors across all constraints and assigns the mutant the minimum such sum across inputs. Unlike prior sum-only methods~\cite{chen2022jigsaw}, this approach is better suited for joint program and input fuzzing. \sys then performs crossover, favoring lower-scoring mutants, and evaluates the resulting offspring on newly sampled inputs. This process repeats until a violation is found or a timeout occurs. Each iteration is independent once crossover mutants are generated, and \sys further improves efficiency via heuristics that prioritize mutation of program statements and input regions most likely to cause bugs.

In our evaluation, we implement \sys for Circom~\cite{belles2022circom}, the most popular programming system for developing ZK circuits~\cite{wen2024practical}, and compare it with four state-of-the-art detection tools on \nallcircuits real-world ZK circuits. Our results show that \sys consistently outperforms all other methods, finding 30\% to 300\% more bugs, without any false positives. In total, \sys found \ndetectedbugs bugs, including \nnewbugs previously unknown zero-days—\nconfirmed of which were confirmed by developers and \nfixed fixed. For example, \nconfirmedzkregex under-constrained bugs in \textit{zk-regex}~\cite{zkregex}, a popular ZK regex verification tool, were confirmed and awarded bug bounties. These bugs allow malicious provers to generate bogus proofs claiming that they possess a string with certain properties (e.g., a valid email address) without actually having such a string. Another confirmed bug in \textit{passport-zk-circuits}~\cite{passport-zk-circuits}, a ZK-based biometric passport project, allowed attackers to forge proofs of possessing data matching a required encoding. We also prototype \sys for Noir~\cite{aztec2024noir}, another emerging DSL, and evaluate it on real Noir circuits. Our code is publicly available at \url{https://github.com/Koukyosyumei/zkFuzz}.

In summary, we make the following key contributions:

\begin{itemize}
  \item Introduce \model, a language-independent formal framework capturing both under- and over-constrained ZK bugs through trace-constraint inconsistencies.
  \item Develop \sys, a fuzzer that detects \model violations by jointly mutating programs and inputs.
  \item Propose a joint evolutionary fuzzing algorithm with a novel min-sum fitness function and guided crossover.
  \item Evaluate \sys on \nallcircuits real-world ZK circuits written in Circom, finding 30--300\% more bugs than prior tools with zero false positives.
  \item Implement a preliminary prototype for Noir, another popular DSL for ZK circuits not supported by existing tools, and demonstrate that \sys can detect a real ZK bug reported previously in Noir circuits.
\end{itemize}

%% file: subfiles/background.tex
\begin{figure*}[!ht]
    \centering
    \includegraphics[width=\linewidth]{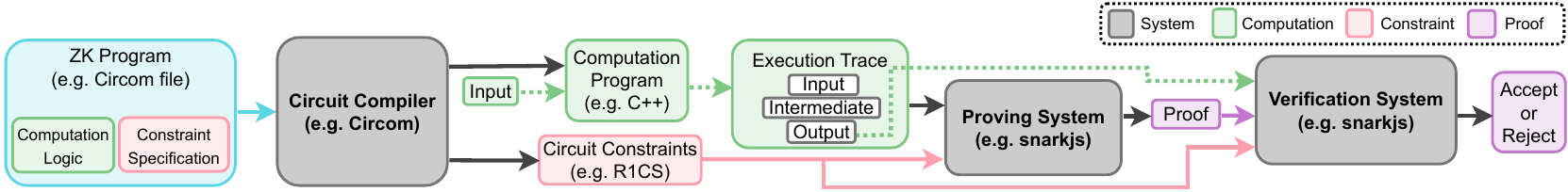}
    \caption{Overview of ZK proof systems. The circuit compiler processes a ZK program into a witness computation program and circuit constraints. Prover executes the witness program to obtain a trace (witness and public values) and creates a proof using the proving system. The verifier validates the proof using the verification system with the public output, the constraint, and the proof. Our fuzzer checks the inconsistencies between the computation logic and the circuit constraints in the ZK program.}
    \label{fig:overview-zkp}
\end{figure*}

\section{Background}
\label{sec:background}

This section provides an overview of ZK Proof systems (\S~\ref{subsec:zkp-system}), bugs in ZK circuits (\S~\ref{subsec:ov}), and Circom, the most popular programming system for ZK circuits~\cite{wen2024practical} (\S~\ref{subsec:circom-lang}).

\subsection{ZK Proof Systems}
\label{subsec:zkp-system}

ZK proofs enable a \textit{prover} to convince a \textit{verifier} of the validity of a statement without revealing any information beyond the statement's truth~\cite{LI201425}. Formally, ZK proof systems satisfy three properties: 

\begin{itemize}
\item \textbf{Completeness}: If the statement is true, an honest prover convinces an honest verifier.
\item \textbf{Soundness}: A malicious prover cannot convince the verifier of a false statement.
\item \textbf{Zero-Knowledge}: The proof reveals nothing about the secret input beyond the statement's validity.
\end{itemize}

\noindent 
%Real-world implementations, however, often contain errors, breaking those theoretical guarantees.
%wrong implementation breaks the security gurantee
These properties enable the construction of privacy-preserving verifiable systems, such as anonymous cryptocurrencies (e.g., Zcash's zk-SNARKs), confidential smart contracts, and verifiable voting protocols. Modern applications extend to blockchain scaling (e.g., ZK-Rollups) and verification of machine-learning integrity~\cite{lavin2024survey,garg2023experimenting,abbaszadeh2024zero}.

At a high level, a ZK system operates on arithmetic circuits that perform computations over a finite field, where all variables and operations (e.g., addition, multiplication) are defined modulo a large prime. It exposes two primitives: $\mathrm{Prove}$ and $\mathrm{Verify}$. The ZK proof process begins with the prover running $\mathrm{Prove}$ to generate a proof: $\pi \gets \mathrm{Prove}(\mathcal{C}, w)$ where $\mathcal{C}$ is the circuit and $w$ the prover's secret value, known as the \textit{witness}. Importantly, the proof $\pi$ is constructed to reveal no information about $w$. Once $\pi$ is generated, the prover sends it to the verifier who then calls $\mathrm{Verify}(\mathcal{C}, \pi)$ which returns either $accept$ or $reject$. For performance reasons, it is often desirable for the proof to be succinct and non-interactive.

\noindent{\textbf{ZK Programming Infrastructure.}} ZK protocols assume a static secret $w$ but, in practice, developers frequently want to make claims about computations, such as proving membership in a Merkle tree for off-chain transactions---where the witness $w$ includes not just a secret value, but also a computation trace (e.g., a Merkle path) that must satisfy the circuit constraints. For efficient proof generation and verification, the constraints are often expressed in quadratic form like in R1CS (Rank-1 Constraint System). Therefore, general computations %like Merkle proof
do not directly map to the circuit constraints.

Manually constructing circuits for such computations and determining which values belong in the witness is complex and error-prone, so developers often use ZK programming systems~\cite{lavin2024survey} like Circom, which allow them to express the witness computation and the circuit constraints in a single ZK program written in a high-level DSL. The DSL compiler then translates the program into a computation program $\mathcal{P}$ and the corresponding set of circuit constraints $\mathcal{C}$. When $\mathcal{P}$ executes, it produces a full execution trace, including not only the input $x$ and output $y$, but also important intermediate values $z$, to match the required constraint format. For example, computing $y = x^4$ typically requires two quadratic constraints: $z = x^2$ and $y = z^2$. The resulting execution trace $(x, z, y)$ becomes the witness (typically, $y$ is public, while $x$ and $z$ are secret; however, developers sometimes make $x$ public as well, depending on the application and desired level of privacy.) Fig.~\ref{fig:overview-zkp} illustrates this process, along with the subsequent use of proving and verification systems like snarkjs~\cite{snarkjs}.

Although some systems like Noir~\cite{aztec2024noir} support automatic generation of constraints from the computations, the resulting constraints are often significantly less efficient, increasing the number of constraints by 3x to 300x~\cite{noir_lang_zk_bench}. {Consequently, many Noir projects use \textit{unsafe} mode, where programmers manually write constraints like Circom, introducing the risk of discrepancy between the computation program and constraints.} In the remainder of our paper, we focus on practical ZK systems that produce efficient constraints. 

\subsection{Overview of Vulnerabilities in ZK Programs}
\label{subsec:ov}

Even with high-level programming systems, general computations often involve non-quadratic operations that cannot be directly mapped to the required constraint format. As a result, it is the developer’s responsibility to ensure that the intended computation $\mathcal{P}$ and the corresponding circuit constraints $\mathcal{C}$ are aligned; any mismatch may lead to incorrect or unverifiable proofs. Unfamiliarity with finite-field arithmetic further complicates matters, as its modular behavior introduces unexpected corner cases. Consequently, ZK programs are notoriously difficult to get right~\cite{pailoor2023automated,chen2024ac4} and may contain the following two classes of vulnerabilities.

\begin{itemize}
\item{\textbf{Under-Constrained Circuits.}} $\mathcal{C}$ is too loose for $\mathcal{P}$, allowing malicious provers to convince verifiers that "I know the input whose corresponding output is $y$," even when they do not. Such vulnerabilities violate the \emph{soundness} property.

\item{\textbf{Over-Constrained Circuit.}} $\mathcal{C}$ is too strict for $\mathcal{P}$, preventing honest provers from generating valid proofs for some correct traces. Such vulnerabilities violate the \emph{completeness} property.
\end{itemize}

%\hideaki{For example, consider the computation shown in
Code~\ref{lst:program} and~\ref{lst:under-constraint} show the computation and constraints, respectively, for an under-constrained vulnerability caught by \sys %\junfeng{Hide: please confirm} \hideaki{(yes, it's correct)} 
in a real-world 1-bit right-shift program from the \textit{circom-monolith} library~\cite{circom-monolith}. Since right shift is not a quadratic operation and cannot be directly encoded as a constraint, the developer introduces an intermediate value \texttt{b}, defined as the difference between \texttt{x} and \texttt{2y}, where \texttt{y = x >> 1}. 

Unfortunately, due to quirks in finite-field arithmetic, multiple assignments to \texttt{(y, b)} can satisfy the constraints for a given input \texttt{x}. For example, in the finite field $\mathbb{F}_{11}$, if \texttt{x = 7}, the expected output is \texttt{y = 3} and \texttt{b = 1}. Yet, \texttt{y = 9}, \texttt{b = 0} also satisfy the constraints, because $7 - 9 \cdot 2 = -11 \equiv 0 \bmod 11$. The circuit is thus under-constrained, allowing an invalid proof.

In contrast, the constraints in Code~\ref{lst:over-constraint} are too strict. They enforce \texttt{b = 0}, thereby rejecting valid traces such as \texttt{\{x: 3, y: 1, b: 1\}}, and preventing an honest prover from generating a valid proof.

\ifpreprint
The safe implementation of 1-bit right shift using the bit array can be found in Appendix~\ref{appendix:examples}.
\fi

\begin{figure}[!ht]
  \centering
  \begin{minipage}[t]{0.37\linewidth}
    \centering
    \lstset{frame=single}
    \begin{lstlisting}[caption={%Computation of
    1-bit R shift}, label={lst:program}, language={rust}]
fn RShift1(x) {
 y = x >> 1
 b = x - y * 2
 assert(b*(1-b) == 0)
 return y
}
    \end{lstlisting}
  \end{minipage}
  \hfill
  \begin{minipage}[t]{0.56\linewidth}
    \centering

    \lstset{frame=single}
\begin{lstlisting}[caption={Under-constrained circuit.}, label={lst:under-constraint}]
(b = x-y*2) && (b*(1-b) = 0)
\end{lstlisting}
    
    \lstset{frame=single}
\begin{lstlisting}[caption={Over-constrained circuit.}, label={lst:over-constraint}]
(b = x - y*2) && (b = 0)
\end{lstlisting}
  \end{minipage}
\end{figure}

\WFclear

\noindent{\textbf{Out-of-Scope Vulnerabilities.}} \sys focuses on soundness and completeness vulnerabilities in ZK proofs that arise from developer mistakes in tying constraints to computation. The following types of vulnerabilities are out of scope: (1) vulnerabilities in the underlying compilation, proving, and verification systems (e.g., the Frozen Heart vulnerability caused by an insecure implementation of the Fiat-Shamir transformation~\cite{tang2024zero}); (2) bugs solely in the computation logic, such as unused inputs in hashing~\cite{wen2024practical} (such unused inputs may be acceptable depending on the use case, as in computing the trace of a matrix where only diagonal elements are used); and (3) bugs solely in ZK protocol design that leak information and violate the zero-knowledge property~\cite{chaliasos2024sok}, inline with all prior detectors~\cite{circomspect,pailoor2023automated,jiang2025conscs,chen2024ac4,wen2024practical}---what constitutes a secret is highly application-dependent: e.g., in Validity Rollups~\cite{chen2022review}, ZK proofs are used to accelerate transactions rather than to conceal them.

\subsection{Circom Primer}
\label{subsec:circom-lang}

While our vulnerability definition and \sys techniques apply to different ZK programming systems, this paper focuses on Circom, the most popular DSL and infrastructure for developing ZK circuits. 
\\
\noindent{\textbf{Value Domains and Operators.}} In Circom, all computations are performed over a finite field $\mathbb{F}_q$, where $q$ is a large prime number. The language supports standard arithmetic operators such as addition ($+$), subtraction ($-$), multiplication ($*$), and division ($/$), all performed modulo $q$. Additionally, Circom provides integer division ($\backslash$), bitwise, and comparison operators for defining complex circuit logic. The operands or results of these operators often need implicit conversion between integer and finite fields.
\\
\noindent{\textbf{Signals and Variables.}} Circom enforces a strict separation between constrained and unconstrained computation through two data types.
% Please add the following required packages to your document preamble:
% \usepackage{booktabs}
%\begin{table}[!ht]
%\begin{tabular}{@{}lll@{}}
%\toprule
%\textbf{Property}     & \textbf{Signal ("signal")}    & \textbf{Variable ("var")} \\ \midrule
%Mutability            & Immutable                     & Mutable                   \\ \midrule
%Constraint Generation & Possible                      & Prohibited                \\ \midrule
%Assignment            & \verb|<--| (unconstrained) & \verb|=| (imperative)            \\
%                      & \verb|<==| (constrained)   &                           \\ \bottomrule
%\end{tabular}
%\caption{Circom distinguishes between signals and variables to enforce strict constraint integrity and better expressiveness. Signals represent immutable values that participate in circuit constraints, whereas variables allow mutable, unconstrained computations.}
%\end{table}
\textit{Signals} are the primary data type and define the flow of data and the constraints that must be satisfied. They are immutable once defined and can serve as input, output, or intermediate values.
%There are three types of signals:
%\begin{itemize}
%    \item \textbf{Input}: Represents an external input to the circuit.
%    \item \textbf{Output}: Represents the result of the circuit's computation.
%    \item \textbf{Intermediate}: Represents a value computed within the circuit but not exposed as an input or output.
%\end{itemize}
%Signals are the backbone of Circom circuits, as they define the flow of data and the constraints that must be satisfied.
In contrast, \textit{variables}, declared using the \texttt{var} keyword, are unconstrained local values used for the computations but are not tracked by the constraint system. Unlike signals, variables are mutable and can be reassigned during execution. %However, variables cannot be used in constraints directly; their values are unfolded into constants or expressions of signals when compiling the Circom file into the computation program and the constraint.
\\
\noindent{\textbf{Constraint and Assignment Semantics.}} Circom constraints follow a quadratic form: the product of two linear expressions equal a third linear expression. To express such constraints alongside computations,  Circom introduces the following operators.

%\begin{itemize}
\textit{Weak Assignment} (\verb|<--|) assigns a value to a signal \emph{without} generating a constraint,  usually used for intermediate computations where constraints are unnecessary.

\textit{Strong Assignment} (\verb|<==|) assigns a value to a signal \emph{and} generates a constraint, ensuring the assigned value is enforced by the circuit.

\textit{Assert} (\verb|assert|) inserts a runtime assertion into the computation \emph{without} adding a constraint.

\textit{Equality Constraint} (\verb|===|) asserts equality between two signals \emph{and} generates a constraint. If the \texttt{constraint\_} \texttt{assert\_} \texttt{disabled} flag is set, the assertion is \emph{not} inserted into the computation.

\textit{Variable Assignment} (\verb|=|) assigns a value to a variable, and does not add a constraint to the circuit.

Other than \verb|===|, Circom operators do not allow adding constraints without performing the corresponding computation, so \verb|===| with \texttt{constraint\_assert\_disabled} set is the only way to introduce over-constrained bugs. Nonetheless, over-constrained circuits can arise in other languages such as zigren~\cite{veridise2025risc0} and halo2~\cite{halo2,soureshjani2023automated}, making it important to define them formally.

\noindent{\textbf{Templates and Components.}} Circom \textit{templates} are parameterized circuit blueprints that can be instantiated as \textit{components} by other circuits. The parameters can be, for example, the dimensions of input arrays, making them highly flexible. Templates and components enable modular design, but also introduce the risk that bugs in one template can propagate across all circuits that instantiate it, potentially affecting many downstream ZK programs.

%% file: subfiles/bug_definition.tex
\section{Definitions of ZK Program Vulnerabilities}
\label{sec:formulation}

This section formulates the typical ZK circuit's vulnerabilities as the \textbf{Trace-Constraint Consistency Test (TCCT)}, a unified {language-agnostic} model that rigorously formulates bugs within ZK circuits as {discrepancies between traces produced by a witness-generation program on possible inputs and the constraints enforced by the circuit.} %\delete{the inconsistency between the computation logic and its associated circuit constraint}.

\subsection{Incompleteness of Prior Definitions}
{Prior formal definitions of ZK program vulnerabilities are incomplete on two key fronts, as we elaborate below. Both gaps stem from prior work’s exclusive focus on constraints while neglecting the computation logic in witness generation.}

First, when detecting under-constrained vulnerabilities, they focus on identifying nondeterministic circuit constraints~\cite{pailoor2023automated,jiang2025conscs}, while ignoring computation aborts. If the circuit constraints $\mathcal{C}$ allow the same input $x$ to map to an output $y'$ different from the computed output $y$, the circuit is clearly under-constrained. However, even if $\mathcal{C}$ is deterministic—accepting only one value of $y$ for a given $x$—the circuit can still be under-constrained if the computation $\mathcal{P}$ aborts on $x$ (e.g., due to a runtime assertion). Attackers can exploit such aborts by crafting inputs that cause the computation to fail while still satisfying the circuit constraints, thereby producing bogus proofs for traces that never actually occurred. Our evaluation shows that such abort vulnerabilities account for nearly 50\% of the real-world bugs \sys identified.

Second, when detecting over-constrained vulnerabilities, prior work~\cite{chen2024ac4} considers only a special case where $\mathcal{C}$ permits an empty set of traces. A complete definition must consider the full set of execution traces produced by $\mathcal{P}$ and those accepted by $\mathcal{C}$. In particular, it must account for intermediate values. Suppose $\mathcal{P}$ produces a trace $(x, z, y)$ for input $x$, but $\mathcal{C}$ accepts only $(x, z', y)$ where $z \neq z'$. In this case, $\mathcal{C}$ is over-constrained with respect to $\mathcal{P}$, as an honest prover with the valid trace $(x, z, y)$ would be unable to generate a valid proof. If, however, $\mathcal{C}$ accepts both $(x, z, y)$ and $(x, z', y)$ (a common optimization in practice), there is no violation, as the prover can produce a valid proof.

\iffalse
\begin{figure}[!ht]
  \centering
  \begin{minipage}[t]{0.45\linewidth}
    \centering
    \lstset{frame=single}
    \begin{lstlisting}[caption={Computation}, label={lst:program-crash}, language={rust}]
fn Verify(x) {
 y = x * 2 + 1
 assert(y == 0)
 return y
}
    \end{lstlisting}
    \lstset{frame=single}
\begin{lstlisting}[caption={Although this circuit is deterministic, it is still under-constrained.}, label={lst:under-constraint-crash}]
y = x * 2 + 1
\end{lstlisting}
  \end{minipage}
  \hfill
  \begin{minipage}[t]{0.47\linewidth}
    \centering
    \lstset{frame=single}
\begin{lstlisting}[language={Circom}, style=circomstyle, caption={Corresponding Circom code. Existing vulnerability detection tools cannot identify this bug.}, label={lst:Example-first}]
template Verify() {
 signal input x;
 signal output y;

 y <== x**2 + 1;
 assert(y == 0);
}
\end{lstlisting}
  \end{minipage}
\end{figure}
\fi

\subsection{Trace-Constraint Consistency Test}

%\paragraph{\textbf{Specifications}} 
Let $\mathbb{F}_q$ be a finite field of a prime order $q$.

\begin{definition}[ZK Program]
A \textit{ZK program} is defined as a pair of a computation and a set of circuit constraints $(\mathcal{P}, \mathcal{C})$:

$\mathcal{P} : \mathbb{F}_q^n \to \mathbb{F}_q^k \times (\mathbb{F}_q^m \cup \{\bot\})$ is a \textit{computation} that takes as input a tuple $(x_1, x_2, \dots, x_n) \in \mathbb{F}_q^n$, and behaves as follows:
  \begin{itemize}
    \item On successful termination, outputs a pair $(z, y)$ where $z = (z_1, \dots, z_k) \in \mathbb{F}_q^k$ represents intermediate values necessary to construct a witness and $y = (y_1, \dots, y_m) \in \mathbb{F}_q^m$  the final outputs. Each value is assigned exactly once.% and remains fixed thereafter.
    \item On abnormal termination, it returns the special value $(\_, \bot)$.
  \end{itemize}

$\mathcal{C} : \mathbb{F}_q^n \times \mathbb{F}_q^k \times \mathbb{F}_q^m \to \{\text{true}, \text{false}\}$ is a set of circuit \textit{constraints}, which evaluates to \texttt{true} if the given input, intermediate, and output values satisfy the required conditions; and \texttt{false} otherwise.
\end{definition}

\begin{definition}[Execution Trace]
    For an input $x \in \mathbb{F}_q^n$ and computation $\mathcal{P}$, if $\mathcal{P}(x) = (z, y)$ where $z \in \mathbb{F}_q^k$ and $y \in \mathbb{F}_q^m$, then the triplet $(x, z, y)$ is called an \textit{execution trace} of $\mathcal{P}$.
\end{definition}

The trace set $\mathcal{T}(\mathcal{P})$ is the set of all possible traces, except aborts, generated by the computation $\mathcal{P}$ while the constraint satisfaction set $\mathcal{S}(\mathcal{C})$ is the set of all tuples satisfying the circuit constraints $\mathcal{C}$.

\begin{definition}[Trace Set]
For a computation $\mathcal{P}$, the trace set $\mathcal{T}(\mathcal{P})$ is defined as:

\begin{align*}
    \mathcal{T}(\mathcal{P}) := \{(x, z, y) \mid &x \in \mathbb{F}^{n}_q, z \in \mathbb{F}^{k}_q, y \in \mathbb{F}^{m}_q, \mathcal{P}(x) = (z, y)\}
\end{align*}
\end{definition}

\begin{definition}[Constraint Satisfaction Set] For circuit constraints $\mathcal{C}$, the constraint satisfaction set $\mathcal{S}(\mathcal{C})$ is defined as:

\begin{align*}
\mathcal{S}(\mathcal{C}) := 
\left\{
  \begin{array}{r|l}
    \vphantom{\begin{array}{l}x \in \mathbb{F}^n_q \\ C(x, z, y) = \text{true} \end{array}} 
    (x, z, y)
    &
    x \in \mathbb{F}^n_q,\ z \in \mathbb{F}^k_q,\ y \in \mathbb{F}^m_q, \\
    &
    C(x, z, y) = \text{true}
  \end{array}
\right\}
\end{align*}
\end{definition}

\noindent In addition, we introduce an operator that projects a set of tuples of input, intermediate, and output values to a set of pairs of input and output values.

\begin{definition}[Projection]
Let $\{(x_1,$$ z_1,$$ y_1),$$ $$ \hdots\}$$ \subseteq 2^{\mathbb{F}^{n}_q \times \mathbb{F}^{k}_q \times \mathbb{F}^{m}_q}$ be a set of tuples of inputs, intermediates and outputs. Then, we define the projection operator $\Pi_{xy}$ as $\Pi_{xy}(\{(x_1, z_1, y_1), \hdots\}) := \{(x_1, y_1), \hdots\}$.
\end{definition}

Given the above building blocks, we formally define under-constrained and over-constrained circuits.

\begin{definition}[Under-Constrained Circuit]
We say that $\mathcal{C}$ is \textit{under-constrained} for $\mathcal{P}$ if

\begin{equation}
\label{eq:uc}
\Pi_{xy}(\mathcal{S}(\mathcal{C})) \setminus \Pi_{xy}(\mathcal{T}(\mathcal{P})) \neq \emptyset
\end{equation}

\end{definition} 

\noindent Intuitively, $\mathcal{C}$ is under-constrained for $\mathcal{P}$ if $\mathcal{C}$ accepts $(x, \_, y)$ but $\mathcal{P}(x) = (\_, y')$ and $y \neq y'$ or $\mathcal{P}(x)$ aborts. Projection out the intermediate values here is necessary for the following reason. In contrast to over-constrained vulnerabilities, if $\mathcal{P}(x) = (z, y)$ and $\mathcal{C}$ accepts $(x, z', y)$ such that $z \neq z'$, it is \emph{not} an under-constrained vulnerability, since malicious provers still cannot create bogus proofs asserting that $\mathcal{P}(x)$ yields an output $y' \neq y$. Similarly, even if $\mathcal{C}$ accepts both $(x, z, y)$ and $(x, z', y)$, adversaries cannot forge proofs with outputs differing from $y$.

This definition naturally takes care of abnormal termination of the computation. Any input $x$ that causes $\mathcal{P}$ to abort will not appear in the trace set $\mathcal{T}(\mathcal{P})$. If $\mathcal{S}(\mathcal{C})$ includes any tuple of the form $(x, \_, \_)$, then the circuit is under-constrained, as it accepts an input that the computation cannot successfully process.

%On the other hand, the over-constrained circuit means that the circuit constraints cannot accept the valid trace of the computation, breaking the second term of Eq.~\ref{tcct-condition}.

\begin{definition}[Over-Constrained Circuit]
We say that $\mathcal{C}$ is \textit{over-constrained} for $\mathcal{P}$ if

\begin{equation}
\label{eq:oc}
\mathcal{T}(\mathcal{P}) \setminus \mathcal{S}(\mathcal{C}) \neq \emptyset    
\end{equation}

\end{definition}

\noindent
Intuitively, $\mathcal{C}$ is over-constrained w.r.t. $\mathcal{P}$ if any valid trace of $\mathcal{P}$ is not accepted by $\mathcal{C}$. For example, if $\mathcal{P}(x)$ yields output $y$, but $\mathcal{C}$ accepts only $y' \neq y$ for the same input $x$, then $\mathcal{C}$ is over-constrained.

This definition must include intermediate values $z$, as they are essential for satisfying $\mathcal{C}$. Given $\mathcal{P}(x) = (z, y)$, consider the two cases discussed in the previous subsection:
(1) $\mathcal{C}$ accepts only $(x, z', y)$ with $z' \neq z$. Since $\mathcal{T}(\mathcal{P}) \setminus \mathcal{S}(\mathcal{C}) \neq \emptyset$, $\mathcal{C}$ is over-constrained.
(2) $\mathcal{C}$ accepts both $(x, z, y)$ and $(x, z', y)$. Since execution trace $(x, z, y)$ occurs in $\mathcal{S}(\mathcal{C})$, this is not an over-constrained vulnerability, even though multiple traces share the same input $x$ and output $y$. In fact, allowing a broader range of intermediate values in $\mathcal{C}$ is a common optimization in practical ZK circuit design.

%\noindent Note that they are not disjoint and may co-occur. \hideaki{While the formulation of under-constrained circuits focuses on the non-deterministic relationship between inputs and outputs, ignoring intermediate variables, we must consider the intermediate variables when modeling over-constrained circuits to formulate the erroneous rejection of valid traces.}

We now combine Eq.~\ref{eq:uc} and Eq.~\ref{eq:oc} to test the consistency between the trace and the constraint satisfaction sets:

\begin{definition}[Trace-Constraint Consistency Test]
%\junfeng{this is just the above two conditions and together. perhaps keep the same operators and also explain this is just two conditions and together} \hideakicomment{(Fixed)}
For computation $\mathcal{P}$ and constraints $\mathcal{C}$, the \textit{Trace-Constraint Consistency Test} (\model) ascertains the following:

\begin{equation}
\label{tcct-condition}
\Pi_{xy}(\mathcal{S}(\mathcal{C})) \setminus \Pi_{xy}(\mathcal{T}(\mathcal{P})) = \emptyset  \hquad \land \hquad \mathcal{T}(\mathcal{P}) \setminus \mathcal{S}(\mathcal{C}) = \emptyset
   %\Pi_{xy}(\mathcal{T}(\mathcal{P})) = \Pi_{xy}(\mathcal{S}(\mathcal{C}))  \hquad \land \hquad \mathcal{T}(\mathcal{P}) \subseteq \mathcal{S}(\mathcal{C}) 
\end{equation}

\noindent If Eq.~\ref{tcct-condition} holds, we say that $\langle \mathcal{P}, \mathcal{C} \rangle$ is \textit{well-constrained}, which is equivalent to $\mathcal{T}(\mathcal{P}) \subseteq \mathcal{S}(\mathcal{C})$ $\land$ $\Pi_{xy}(\mathcal{T}(\mathcal{P})) = \Pi_{xy}(\mathcal{S}(\mathcal{C}))$, meaning that all execution traces satisfy the constraints, and the set of input-output pairs is identical between the execution traces and the constraint satisfaction set.

\end{definition}

%\noindent When there are no intermediate values, Eq.~\ref{tcct-condition} is equivalent to $\mathcal{T}(\mathcal{P}) = \mathcal{S}(\mathcal{C})$. 

%\paragraph{Theoretical Analysis}

The computational complexity of \model is as follows:

\begin{theorem}
\label{thm:under-constrained-detection-np-complete}
Let $\langle \mathcal{P}, \mathcal{C} \rangle$ denote a \model instance.
%\begin{itemize}

\noindent a) Deciding if $\mathcal{C}$ is under-constrained for $\mathcal{P}$  is NP-complete.

\noindent b) Deciding if $\mathcal{C}$ is over-constrained for $\mathcal{P}$  is NP-complete.

\noindent c) TCCT is co-NP-complete.
%\end{itemize}
 
\end{theorem}

\noindent Proof in Appendix~\ref{appendix:proof} is based on the reduction from the Boolean Satisfiability problem and the tautology problem.

\noindent{\textbf{Comparison to Prior Definitions.}} Tab.~\ref{tab:model} compares \model with prior definitions. For under-constrained bugs, prior work detect only non-deterministic behaviors~\cite{pailoor2023automated,jiang2025conscs}, and fail to capture those caused by computation aborts. {In addition, they yield false positives when the expected behavior of the witness generation program is non-deterministic, while our TCCT can properly handle those cases (see Appendix~\ref{appendix:non-deterministic-witness} for a detailed discussion).} For over-constrained bugs, prior work considers only a special case where the constraint satisfaction set is empty~\cite{chen2024ac4}. In contrast, \model is more general and complete, capturing all bugs detectable by existing models, as well as additional bugs that prior definitions overlook. A key advantage of \model is its DSL-independence: even though different DSLs may offer varied primitives or syntactical sugar for specifying constraints and computations in ZK circuits, \model’s definitions consistently apply across all of them. 

% Please add the following required packages to your document preamble:
% \usepackage{booktabs}
\begin{table}[!th]
\caption{Comparison of formal definitions of ZK vulnerabilities. Our \model is more general than existing models.}
\label{tab:model}
\begin{tabular}{@{}l|c|c@{}}
\toprule
Model         & Under-Constrained & Over-Constrained \\ \midrule
\begin{tabular}[l]{@{}l@{}} \cite{pailoor2023automated,jiang2025conscs} \end{tabular} & \begin{tabular}[l]{@{}l@{}}$\exists{x, z, z', y, y'}.$ \\ $ y \neq y' \land $ $ C(x, z, y) \land  C(x, z', y')$\end{tabular}                 &  - \\ \midrule
\cite{chen2024ac4}           &  Same as \cite{pailoor2023automated,jiang2025conscs}                  & $\mathcal{S}(\mathcal{C}) = \emptyset$ \\ \midrule
\textbf{\model}   & $\Pi_{xy}(\mathcal{S}(\mathcal{C})) \setminus \Pi_{xy}(\mathcal{T}(\mathcal{P})) \neq \emptyset$ & $\mathcal{T}(\mathcal{P}) \setminus \mathcal{S}(\mathcal{C}) \neq \emptyset    $  \\ \bottomrule
\end{tabular}
\end{table}

\subsection{Vulnerable and Benign Examples}
\label{subseq:examples}

This subsection shows two real-world Circom circuits, an under-constrained circuit that illustrates our definition, particularly the role of computation aborts; and a well-constrained circuit that highlight the importance of acounting for intermediate values.

\iftrue
\begin{figure}[!hbt]
\begin{minipage}[t]{0.465\linewidth}
\begin{lstlisting}[language=Circom, style=circomstyle, caption={Under-constrained bug from~\cite{tokamakzvm} undetectable by any existing formulations, only captured by our \model.}, label={lst:gen}]
template Transfer(){
 signal input fb, tb;
 signal input amt;
 signal output fn, tn
 assert(fb - amt >= 0);
 fn <== fb - amt;
 tn <== tb + amt;
}
\end{lstlisting}
\end{minipage}
\hfill
\begin{minipage}[t]{0.478\linewidth}
\begin{lstlisting}[language={Circom}, style=circomstyle, caption={Benign circuit from~\cite{circomlib} incorrectly flagged by prior static tools. Its trace and constraint satisfaction sets are shown in Tab.~\ref{tab:trace-cs-sets-well-iszero}.}, label={lst:Example}]
template IsZero() {
 signal input x;
 signal output y;
 signal z;
 z <-- x != 0 ? 1/x : 0;
 y <== -x*z + 1;
 x*y === 0;
}
\end{lstlisting}
\end{minipage}

\begin{minipage}[t]{0.44\linewidth}
\iffalse
\begin{lstlisting}[language=Circom, style=circomstyle,  caption={\texttt{LessThan} implicitly assumes that the bit length of both inputs are not longer than \texttt{n}, although \texttt{LessThan} itself does not check the range of inputs.}, label={lst:dsk}]
template EmailAddrRegex(n) {
 signal input msg[n];

 signal in_range_checks[n];
 for (var i = 0; i < n; i++) {
  in_range_checks[i] <== LessThan(8)([msg[i], 255]);
  in_range_checks[i] === 1;
  // the rest is omitted
\end{lstlisting}
\fi

\end{minipage}
\end{figure}
\fi

%For instance, 
Code~\ref{lst:gen} shows a real-world under-constrained circuit (slightly modified for clarity) caught by \sys from \textit{zvm}~\cite{tokamakzvm}. it had previously gone undetected until our discovery, later confirmed by the developers. The circuit takes as inputs the sender’s balance \verb|fb|, the recipient’s balance \verb|tb|, and the transfer amount \verb|amt|, and produces updated balances \verb|fn| and \verb|tn|. This code tries to ensure that \verb|fb| is greater than \verb|amt| via \verb|assert|. (line 5). Circom's \verb|assert|, however, performs runtime checks only without adding any constraint. As a result, the computation aborts when \verb|fb| is higher than  \verb|amt|, the constraints remain satisfiable for any combination of \verb|fb| and \verb|amt|. No prior formal tools can detect this bug because this circuit is always deterministic. Existing static tools~\cite{wen2024practical,circomspect} also do not cover this misuse of \verb|assert|. To correct this circuit, one must introduce a genuine constraint enforcing \verb|fb > amt|—for example, by using the \verb|LessThan| component from \cite{circomlib}.  \S\ref{subseq:previously-unknown-bugs} also presents a similar bug involving the misuse of \verb|assert|, which is not detected by any existing tools.

% While a static tool~\cite{wen2024practical} does flag this code, it relies on pattern matching and produces many false positives. The fact that this bug went unfixed indicates that suggests that developers who may have run the tool did not inspect all of its warnings.

Code~\ref{lst:Example} shows a benign circuit template from Circom’s de facto standard library, \textit{circomlib}~\cite{circomlib}, which returns one if the input $x$ is zero and zero otherwise. Line 6 introduces an intermediate value $z$ using a weak assignment, which does not add any constraint since it is not a quadratic expression. Line 7 assigns the output $y$ and introduces a constraint relating $x$, $z$, and $y$. Line 8 adds both a runtime assertion

\begin{wraptable}{r}{0.43\linewidth}
\caption{Trace and Constraint Satisfaction Sets of \texttt{IsZero} in Code~\ref{lst:Example} ($q = 3$)}
\label{tab:trace-cs-sets-well-iszero}
\begin{tabular}{@{}cccccc@{}}
\toprule
\multicolumn{3}{c}{$\mathcal{T}(\mathcal{P})$}               & \multicolumn{3}{c}{$\mathcal{S}(\mathcal{C})$} \\ \midrule
x & z & \multicolumn{1}{c|}{y} & x    & z   & y   \\ \midrule
0  & 0  & \multicolumn{1}{c|}{1}   & 0     & 0     & 1     \\
-  & -  & \multicolumn{1}{c|}{-}   & 0     & 1     & 1     \\
-  & -  & \multicolumn{1}{c|}{-}   & 0     & 2     & 1     \\
1  & 1  & \multicolumn{1}{c|}{0}   & 1     & 1     & 0     \\
2  & 2  & \multicolumn{1}{c|}{0}   & 2     & 2     & 0     \\ \bottomrule
\end{tabular}
\end{wraptable}

\noindent and an equality constraint. The full execution trace and constraint satisfaction sets are shown in Tab.~\ref{tab:trace-cs-sets-well-iszero}. Notably, $\mathcal{T}(\mathcal{P})$ and $\mathcal{S}(\mathcal{C})$ are not identical: $\mathcal{S}(\mathcal{C})$ includes two additional tuples. However, \model correctly handles this case. After projecting out intermediates, the sets of $(x, y)$ pairs are identical, indicating no under-constrained bug. While $\mathcal{S}(\mathcal{C})$ includes three different intermediate values, $z = 0, 1, 2$, for the case where $x = 0$ and $y = 1$, it includes the valid tuple $(0, 0, 1)$ from $\mathcal{T}(\mathcal{P})$, so there is no over-constrained bug either. Despite being benign, this circuit is incorrectly flagged by prior static analyzers, which rely on pattern matching rather than a precise semantic definition of ZK bugs.

We refer the readers to Appendix~\ref{appendix:examples} for more examples, including under-constrained circuits caused by insecure use of the \texttt{LessThan} template, a bug that only \sys can detect, and over-constrained circuits.

\iffalse
\begin{table}[!th]
    \begin{minipage}[t]{0.47\linewidth}
    \begin{lstlisting}[language=Circom, style=circomstyle,  caption={Example of over-constrained circuit. Suppose \texttt{assert} is not added in the computation for \texttt{===}.}, label={lst:over-div}]
template SplitReward() {
 signal input x;    
 signal z;
 signal output y;  

 z <-- x \ 2;
 z * 2 === x;
 y <== z + 1;
}
    \end{lstlisting}
    \end{minipage}
    \hfill
        \centering
    \begin{minipage}[t]{0.48\linewidth}
        \centering
        \caption{Trace and Constraint Satisfaction Sets of Over-Constrained Circuit ($q=5$).}
        \label{tab:over-div}
        \begin{tabular}{cccccc}
        \hline
        \multicolumn{3}{c}{$\mathcal{T}(\mathcal{P})$}         & \multicolumn{3}{c}{$\mathcal{S}(\mathcal{C})$} \\ \hline
        x & z & \multicolumn{1}{c|}{y} & x    &  z  & y       \\ \hline
        0 & 0 & \multicolumn{1}{c|}{1}   & 0      &  0 & 1         \\
        \textcolor{red}{1} & \textcolor{red}{0} & \multicolumn{1}{c|}{\textcolor{red}{1}}   & \textcolor{red}{1}    &     \textcolor{red}{3} &  \textcolor{red}{4}        \\
        2 & 1 & \multicolumn{1}{c|}{2}   & 2     &  1  & 2         \\
        \textcolor{red}{3} & {\textcolor{red}{1}} & \multicolumn{1}{c|}{\textcolor{red}{2}}   & \textcolor{red}{3}    &  \textcolor{red}{4}   & \textcolor{red}{0}         \\
        4 & 2 & \multicolumn{1}{c|}{3}   & 4   &  2    & 3         \\ \bottomrule
        \end{tabular}
    \end{minipage}
\end{table}
\fi

%% file: subfiles/design.tex
\section{\sys: Design}
\label{sec:mutation-testing}

This section introduces our novel fuzzing framework, \sys, which detects bugs in ZK circuits by jointly optimizing program mutation and input generation, guided by target selectors. Fig.~\ref{fig:zkfuzz-workflow} provides an overview of the workflow. The design is DSL-agnostic, with an implementation for Circom described in \S~\ref{sec:implementation}.

\subsection{Joint Program Mutation \& Input Generation}

\sys uses both program mutation and input fuzzing for detecting bugs. Although program mutation is widely used in mutation testing to assess test case effectiveness~\cite{papadakis2019mutation}, our objective differs: we leverage mutation specifically to generate alternative traces that continue to satisfy the constraint, exposing undesired behaviors in ZK circuits.

Algo.~\ref{algo-mutation-testing} describes the key steps that examine both $\mathcal{P}$ and its associated constraints $\mathcal{C}$ by using two testing modes:
\\  
%\begin{itemize}
\noindent\textbf{Detection of Over-Constrained Circuits.}  
We first generate input data and execute the computation $\mathcal{P}$. If it completes without crashing, yet the resulting trace fails to satisfy the constraints $\mathcal{C}$, this circuit is \textit{over-constrained}.
    
\noindent\textbf{Detection of Under-Constrained Circuits.}
To uncover under-constrained circuits, we construct a mutated version of the computation program, denoted by $\mathcal{P}'$. For example, the mutation algorithm can randomly substitute one operator for another. If $\mathcal{P}'$ produces a non-crash trace whose output differs from $\mathcal{P}$ while still meeting the circuit constraint, then the constraint is insufficient to enforce correctness—indicating an \textit{under-constrained} circuit. Note that if the original computation $\mathcal{P}$ crashes on the input while the mutated computation $\mathcal{P}'$ successfully generates an execution trace that satisfies the constraint, the original output $y$ is $\bot$. Since $\bot \neq y'$, this case is naturally captured by \sys.

\begin{algorithm}[!ht]
\caption{Core Logic of \sys}
\label{algo-mutation-testing}
\begin{algorithmic}[1]
\For{$i \leftarrow$ $1,2,\cdots\mathrm{MAX\_GENERATION}$}
%\If{$i \bmod %\mathrm{INPUT\_UPDATE\_INTERVAL} = %0$}
\State Generate input data $ x $.
%\EndIf
\State Mutate $ \mathcal{P} $ to $ \mathcal{P}' $
\State Execute both $ \mathcal{P} $ and $ \mathcal{P}' $ on $ x $:
\State \hspace{1.5em} $ \mathcal{P}(x) = (z, y) $, $\mathcal{P}'(x) = (z', y') $.
\If{$ y \neq \bot $ \textbf{and} $ \mathcal{C}(x, z, y) = \text{false} $}
    \State Report "Over-Constrained Problem."
\EndIf
\If{$y' \neq \bot$ \textbf{and} $ y \neq y' $ \textbf{and} $\mathcal{C}(x, z', y') = \text{true} $}
    \State Report "Under-Constrained Problem."
\EndIf
\EndFor
\end{algorithmic}
\end{algorithm}

For simplicity, Algo.~\ref{algo-mutation-testing} shows one input and one program mutant sampled per iteration. Both input sampling and program mutation can be biased toward targets that are more likely to expose vulnerabilities. The algorithm also requires minimal bookkeeping across iterations: once program mutants are generated, the current iteration operates independently of the previous ones.

%should be vertically hiroku
\begin{figure*}[!th]
    \centering
    \includegraphics[width=0.99\linewidth]{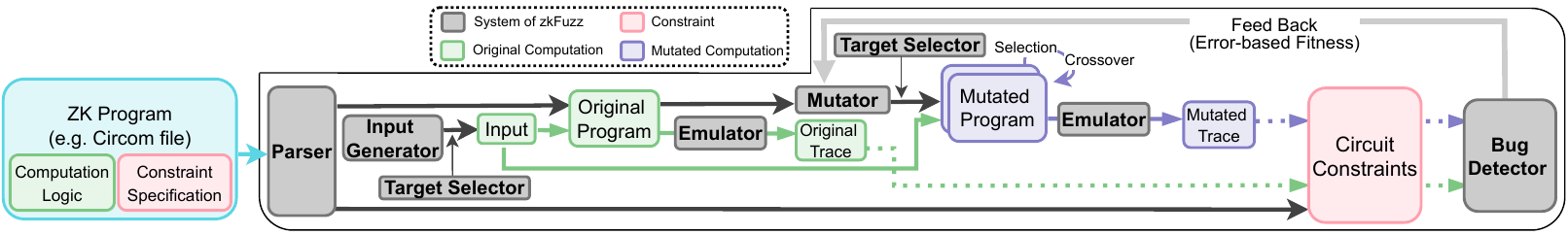}
    \caption{Basic workflow of \textbf{\sys}. This fuzzing framework systematically mutates program logic and feeds artificially generated input data to the original and mutated programs to catch inconsistencies between the program and the constraint. The error-based fitness function is utilized to steer the selection and crossover of mutants. Several target selectors are also applied for input generation and program mutation to guide the search.}
    \label{fig:zkfuzz-workflow}
\end{figure*}

\subsection{Guided Search}
\label{subsec:guided-search}

Since the search space of the program mutation and input generation is huge, we adopt a genetic algorithm with a novel fitness score and target selectors to achieve scalable bug detection.
\\
\noindent{\textbf{Genetic Algorithm with Error-based Fitness Score.}} Genetic algorithm (GA) is a search heuristic commonly used in fuzzing, where candidate solutions are mutated over successive generations, with higher-performing candidates prioritized to guide the search.

\sys uses the smoothed error function transformed from the constraints $\mathcal{C}$ as the fitness score to find an execution trace that satisfies the constraints while producing a different output from the original. This means that it guides the mutation towards under-constrained circuits, and we adopt this design since under-constrained circuits tend to cause more critical outcomes than over-constrained circuits~\cite{chen2024ac4}.

The constraints $\mathcal{C}$ usually consist of a set of polynomial equalities, $\bigwedge_{i=1}^{|\mathcal{C}|} (a_i == b_i)$, where $a_i$ and $b_i$ are quadratic expressions, and $|\mathcal{C}|$ denotes the number of those equations. Following~\cite{chen2022jigsaw}, we convert each equality into an error term $|a_i - b_i|$, which evaluates to 0 when the constraint is satisfied and a positive value when violated. The total error for a given trace is defined as the sum of these errors if the trace produces a different output from the original. Otherwise, the error is set to $\infty$, reflecting the failure to find a diverging execution. The final fitness score of a mutated program is determined by evaluating its error over multiple inputs and selecting the minimum score. Our \textit{min-sum} method differs from prior max-only method~\cite{chen2022jigsaw} because \sys jointly mutates both programs and inputs.

In each iteration, \sys uses roulette-wheel selection~\cite{razali2011genetic}, a widely adopted method in genetic algorithms, to generate new program mutants or perform crossover between mutants from the previous iteration. The selection is biased toward mutants with lower scores, as lower values indicate higher fitness. \sys represents a program mutant as a map from candidate locations to the actual mutations applied, so crossover is implemented as a straightforward merging of the maps from the two parent mutants.

\sys can also optionally guide input generation via the above error-based fitness score, where the final fitness score of an input is defined as the minimum score over multiple program mutations. %(see Appendix~\ref{appendix-additional-experiments} for its experiment).
\\
\noindent{\textbf{Target Selectors.}} We further improve the search of \sys by incorporating target selectors, which guide both input generation and program mutation. Unlike traditional targeted fuzzers—which primarily detect issues like memory access violations by finding inputs that reach specific code blocks~\cite{weissberg2024sok}—our target selectors are designed to maximize the odds of uncovering discrepancies between the computation program $\mathcal{P}$ and the circuit constraints $\mathcal{C}$.

The default target selector of \sys is \textit{skewed distribution}, which generates input data randomly while biasing the distribution to regions more likely to contain bugs. Specifically, empirical analysis reveals that many bugs stem from edge-case inputs, particularly values such as $0$ and the prime modulus $q$. To effectively target these cases, \sys samples inputs and substitutes constants in unconstrained assignments using a skewed distribution that increases the likelihood of uncovering hidden vulnerabilities. Specifically, values are sampled from predefined ranges with the following probabilities: binary values ($0$ and $1$) with 15\%, small positive integers ($2$ to $10$) with 34\%, large values near the field order ($q - 100$ to $q-1$) with 50\%, and all other values with 1\%. By biasing input selection toward these critical edge cases, our approach enhances the detection of subtle bugs that might otherwise go unnoticed. The impact of the choice of skewed distribution can be found in \S~\ref{subsec:ablation}. 
%\junfeng{hard to say that you use domain knowledge but not the bugs to select the range}

In addition, our \sys optionally leverages static analysis to mutate computation programs and generate inputs more likely to trigger unexpected behaviors. For example, one of the main root causes of real-world under-constrained circuits is the \textit{zero-division} pattern, wherein the division's numerator and denominator may be zero~\cite{wen2024practical}. In this scenario, a naive constraint for computing $y = x_1/x_2$ might be written as $y \cdot x_2 = x_1$. However, if both $x_1$ and $x_2$ are zero, any value of $y$ will satisfy this constraint, leaving the circuit under-constrained. Therefore, when this pattern is detected, increasing the chance that generated values for $x_1$ and $x_2$ are zero can help reveal these vulnerabilities. Details of these static-analysis-based target selectors are provided in \S~\ref{subsec:heuristic-static-analyzers}.

%% file: subfiles/implementation.tex
\section{\sys: Implementation}
\label{sec:implementation}

We implemented \sys for Circom in about 4,500 lines of Rust. The core implementation follows Algo.~\ref{algo-mutation-testing} with minimal heuristics. On top of this, we develop \sysp, which extends the core system with static analysis-based target selectors to guide the mutation search more efficiently.

\subsection{\sysb: Core Implementation}

\sys begins with an enhanced parser that decomposes Circom programs into an Abstract Syntax Tree (AST) for the computation $\mathcal{P}$ and corresponding circuit constraints $\mathcal{C}$. A mutator generates program variants ($\mathcal{P}'$), and an emulator produces execution traces for these variants using artificial inputs. A bug detector checks if these traces satisfy the original circuit constraints based on \model. Target selectors strategically guide input generation and program mutation towards likely bugs. Finally, mutants' fitness scores drive selection and crossover in subsequent generations (Fig.~\ref{fig:zkfuzz-workflow}). 

\noindent{\textbf{Circom Program Mutation}.} \sys mutates Circom programs using the following methods.

(a) \textit{Assertion Removal}:  
We begin by removing all \verb|assert| statements from $\mathcal{P}$ to reduce the likelihood of abnormal termination in the mutated program $\mathcal{P}'$.

(b) \textit{Assignment Transformations}:  
We randomly alter the right-hand side of assignment operators. Based on the popular mutation strategies used in mutation testing~\cite{papadakis2019mutation}, we apply two types of mutations:
%\begin{itemize}
1) By default, we replace the right-hand expression with a random value over the finite field sampled from the skewed distribution (see \S~\ref{subsec:guided-search});
2) If the right-hand side contains an operator, we may also substitute it with a different operator with the same category, where we categorize the operators into arithmetic, bitwise, and logical operators.
%\end{itemize}

We also implement two optional mutation methods, random value addition and expression deletion. We show detailed results in Appendix~\ref{appendix:optionalmutation}.

\noindent{\textbf{Guided Search}.} \sys guides the search by default using an error-based fitness function and a skewed distribution, as described in \S~\ref{subsec:guided-search}.

\subsection{\sysp: Optimization with Target Selection via Static Analyzers}
\label{subsec:heuristic-static-analyzers}

As a further optional optimization, \sysp leverages static analysis to identify promising targets, guiding both program mutation and input generation to navigate the vast search space of ZK circuits.

\begin{table}[!ht]
\caption{Summary of static target selectors}
\label{tab:summary-static-selector}
\begin{tabular}{@{}lll@{}}
\toprule
Target Selector               & Type                                                                             & Pattern                  \\ \midrule
%Skewed Distribution     & \begin{tabular}[c]{@{}l@{}}Prog. Mutation \& \\ Input Generation\end{tabular} &  (Default)                        \\
%\hline
%\hline
Weak Assignment         & Prog, Mutation & \verb|<--| \\
White List              & Prog. Mutation                                                                & \verb|IsZero|, \verb|Num2Bits|         \\
Zero-Division           & Input Generation                                                                 & \verb|y <-- x1/x2|  \\ 
Invalid Array Subscript & Input Generation                                                                 & \verb|x[too_big_val]| \\
%Binary-Check            & Input Generation                                                                 & \verb|x*(x-1) === 0|        \\
Hash-Check              & Input Generation                                                                 & \verb|h(x1) === x2|             \\
\bottomrule
\end{tabular}
\end{table}

\noindent{\textbf{Weak Assignment Mutation.}} \sysp focuses on the weak assignment (\verb|<--|) for the assignment transformation mutation, since \verb|<--| does not introduce constraints, and modifying its right-hand side does not affect the constraints. This optimization is based on the observation that mutating the strong assignment (\verb|<==|) is ineffective, as it will likely violate the associated condition.

\noindent{\textbf{White List.}} Existing research and public audit reports have formally proven the safety of certain commonly used circuits, even if they contain weak assignments. Since mutating these circuits would immediately lead to constraint violations, \sysp allows specifying a \textit{white list} of circuit templates to be skipped during the mutation phase, concentrating the mutation efforts on less well-understood or more error-prone components. In this study, we include the two most commonly used \textit{circomlib} circuits with weak assignments, \verb|IsZero| and \verb|Num2Bits|, in the white list.

\noindent{\textbf{Zero-Division Pattern.}} As shown in \S~\ref{subsec:guided-search}, the \textit{zero-division} pattern~\cite{wen2024practical} is a common source of under-constrained circuits where both the division's numerator and denominator can be zero. 

Code~\ref{lst:zd} illustrates a simple case of this issue: when both numerator and denominator are zero, the multiplication constraint (see line 7) is trivially satisfied for any value of \verb|y|, due to missing checks for zero-division.

To find this problem, we introduce a target selector that specifically seeks input combinations setting both the numerator and the denominator to zero. When a numerator or denominator reduces to a polynomial of degree at most two in an input variable, we solve the simplified quadratic equation over finite fields to generate suitable inputs. More concretely, our approach is as follows: 

(Degree 1) Consider the expression $x + a \equiv 0 \bmod{q}$, where $x$ denotes the input variable and $a$ is the remaining terms. The solution is then given by $x \equiv -a \bmod{q}$;

(Degree 2) Consider the quadratic expression $a x^2 + bx + c \equiv 0 \bmod{q}$. A solution for $x$ is given by $x \equiv \frac{-b + \sqrt{b^2 - 4ac}}{2a} \bmod{q}$, provided that a square root of $b^2 - 4ac$ exists modulo $q$. Its existence can be verified using Euler's criterion, and the Tonelli-Shanks algorithm~\cite{tonelli1891bemerkung,shanks1972five} is then employed to compute $\sqrt{b^2 - 4ac}$, yielding the solution. 

If coefficients, $a$, $b$, and $c$ in the above, involve additional input variables, their values are determined by substituting artificially generated values.

\noindent{\textbf{Hash-Check Pattern.}} A common pattern in ZK circuit is the \emph{hash-check} pattern, where a circuit verifies that the hash of some input data matches a provided hash value (see Code~\ref{lst:hash-check}). Naive fuzzing with program mutation may struggle to generate valid input-hash pairs, especially when cryptographic hash functions are involved. To overcome this limitation, we employ a heuristic: with a certain probability if an equality constraint (e.g., \verb|A == B|) is encountered where \verb|A| is an input signal that has not been assigned at this point, we update the assignment of \verb|A| to match the value of \verb|B|. This adjustment increases the likelihood that the input data will satisfy the hash-check constraint. A real-world example is shown in Appendix~\ref{appendix:examples}. 

\begin{figure}[!ht]
\begin{minipage}{0.46\linewidth}
    \begin{lstlisting}[language=Circom, style=circomstyle, caption={Example of Zero Division pattern. If both \texttt{x1} and \texttt{x2} are 0, any value for \texttt{y} can satisfy the constraint.}, label={lst:zd}]
template ZeroDiv() {
  signal input x1, x2;
  signal output y;
  y <-- x1 / x2;
  y * x2 === x1;
}
    \end{lstlisting}
    \end{minipage}
\hfill
\begin{minipage}{0.47\linewidth}
    \begin{lstlisting}[language=Circom, style=circomstyle, caption={Example of Hash Check pattern. Finding \texttt{x1} whose hash is \texttt{x} requires inversing the hash function.}, label={lst:hash-check}]
template HashCheck() {
  signal input x1, x2;
  component h = Hash();
  h.x <== x1;
  x2 === h.y;
}
    \end{lstlisting}
    \end{minipage}
\end{figure}

\noindent{\textbf{Invalid Array Subscript.}} Circom-generated programs may crash due to assert violations or out-of-range array indices. Although our mutation algorithm removes assertions, crashes can still occur from invalid array subscripts. To address this, we monitor for repeated out-of-range errors and, when detected, shrink the maximum value of generated input to the minimum value of array dimensions in that program. This constraint can reduce the chance of triggering such crashes, ensuring the fuzzing remains stable.

Collectively, these target selectors enable \sysp to navigate the enormous search space inherent in ZK circuits efficiently. By strategically guiding both program mutation and input generation, they significantly enhance our ability to detect critical vulnerabilities.

\subsection{Support for Other DSLs}

Although our current implementation primarily focuses on Circom, the core design of \sys as presented in Algo \ref{algo-mutation-testing} remains agnostic to any particular language, and we have developed an initial prototype for Noir \cite{aztec2024noir}, a next-generation ZK DSL gaining traction. The details of the implementation and the evaluation can be found in \S~\ref{subseq:noir-evaluation}.

%% file: subfiles/evaluation.tex
\section{Evaluation}
\label{sec:evaluation}

This section presents a comprehensive evaluation designed to address the following research questions:

%\begin{itemize}
 \noindent \textbf{RQ1:} How effective is \sys at detecting vulnerabilities in real-world zero-knowledge (ZK) circuits?
 
 \noindent \textbf{RQ2:} How quickly can \sys detect bugs?
 
 \noindent \textbf{RQ3:} Can \sys uncover previously unknown bugs in production-grade ZK circuits?  
 
    %\item \textbf{RQ3:} What is the individual contribution of each heuristic and static analyzer in \textbf{\sys} to efficient bug detection?  
 \noindent \textbf{RQ4:} What is the individual contribution of each heuristic and static analyzer in \sysp to efficient bug detection?
 
 \noindent \textbf{RQ5:} How sensitive is \sys's performance?

\noindent \textbf{RQ6:} Is \sys applicable for other DSLs like Noir?
 
%\end{itemize}

Appendix~\ref{appendix-additional-experiments} further considers the following questions:

 \noindent \textbf{RQ7:} Is there an alternative mutation strategy?

 \noindent \textbf{RQ8:} How effective are the other baselines (the default version of ZKAP and circomspect with recursive analysis)?

 \noindent \textbf{RQ9:} What are the real-world impacts of previously unknown bugs discovered by \sys?
 
 \noindent \textbf{RQ10:} How many over-constrained circuits occur when the \texttt{constraint\_assert\_disabled} flag is set?
%\end{itemize}

\noindent{\textbf{Benchmark Datasets.}} We evaluate \sys on \nallcircuits real-world test suites written in Circom. Our benchmark extends the dataset from ZKAP~\cite{wen2024practical} by adding new test cases from 40 additional projects, chosen based on the number of GitHub stars, forks, dependent projects, and the impact on active production services. Following~\cite{pailoor2023automated}, we classify the circuits into four categories based on the number of clauses in the constraint, denoted as $|\mathcal{C}|$. The categories are defined as follows: \textit{Small} ($|\mathcal{C}| < 100$);
\textit{Medium} ($100 \leq |\mathcal{C}| < 1000$);
\textit{Large} ($1000 \leq |\mathcal{C}| < 10000$);
\textit{Very Large:} ($10000 \leq |\mathcal{C}|$). 

\noindent{\textbf{Configurations.}} We set a timeout of 2 hour and a maximum of 50000 generations for all our tests. We also set the number of program mutants per generation to 30, each paired with generated inputs of the same size. Program mutation and crossover operations are applied with probabilities of 0.3 and 0.5, respectively. By default, the right-hand side of a weak assignment is replaced with a random constant. However, if it consists of an operator expression, the operator is randomly substituted with a probability of 0.1. When a zero-division pattern is detected, we analytically solve for the variable and substitute it with the computed solution with a probability of 0.2. Those parameters are determined with reference to prior works on genetic algorithm mutation testing~\cite{shin2019empirical,hassanat2019choosing,lin2003adapting}, and we empirically demonstrate the robustness of \sys across different hyperparameter settings in \S~\ref{subseq:sensitivity}. We run \sys with five different random seeds. All experiments were run on an Intel(R) Xeon(R) CPU @ 2.20GHz and 31GB of RAM, running Ubuntu 22.04.3 LTS.

\noindent{\textbf{Baseline Approaches.}} We compare \sys against four state-of-the-art Circom bug detection tools: Circomspect~\cite{circomspect}, ZKAP~\cite{wen2024practical}, Picus~\cite{pailoor2023automated}, and ConsCS~\cite{jiang2025conscs}. We use two SMT solvers, Z3~\cite{de2008z3} and CVC5~\cite{barbosa2022cvc5} for Picus.  To ensure a fair comparison, we do not count warnings related to \verb|IsZero| and \verb|Num2Bits| as false positives, since \sys explicitly whitelists these templates. Additionally, ZKAP’s \textit{unconstrained signal (US)} detector flags any circuit containing unused signals as unsafe. Although the original ZKAP paper considers unused inputs problematic, our work does not regard them as bugs (see \S~\ref{sec:rw}). Moreover, as noted in~\cite{wen2024practical}, the primary source of ZKAP’s false positives is its \textit{unconstrained sub-circuit output (USCO)} check, which targets unused component outputs. Therefore, we exclude both US and USCO from our main results to improve ZKAP’s precision. Note that the original Circomspect does not analyze internally called templates recursively. 

\begin{table*}[!ht]
\caption{The number of unique bugs detected by each tool, categorized by circuit size (measured by the size of constraint). TP and FP denote true positive and false positive, respectively. Prec. stands for precision, which is computed as TP / (TP + FP). Max / Min refers to the highest and lowest number of detected bugs across five different random seeds. The magnitude of TP and Precision, and FP are highlighted in \colorbox{teal!50}{blue} and \colorbox{red!50}{red}, respectively. Values in the lower row of each size category, colored with \colorbox{gray!50}{gray}, indicate the results only on the ZKAP dataset~\cite{wen2024practical}. \sys outperforms all existing methods.}
\label{tab:main-result}
\centering
\begin{adjustbox}{width=\textwidth,center}
\begin{tabular}{lcc|ccc|ccc|ccc|ccc|ccc|ccc}
\toprule
\begin{tabular}[c]{@{}l@{}}Constraint \\ Size\end{tabular} & \#Bench & \#Bug & \multicolumn{3}{c|}{Circomspect~\cite{circomspect}} & \multicolumn{3}{c|}{ZKAP~\cite{wen2024practical}} & \multicolumn{3}{c|}{\begin{tabular}[c]{@{}c@{}}Picus~\cite{pailoor2023automated} \\ (Z3 / CVC5)\end{tabular}} & \multicolumn{3}{c|}{\begin{tabular}[c]{@{}c@{}}ConsCS~\cite{jiang2025conscs}\end{tabular}} & \multicolumn{3}{c|}{\begin{tabular}[c]{@{}c@{}}\textbf{\sysb} \\ (Max / Min)\end{tabular}} & \multicolumn{3}{c}{\begin{tabular}[c]{@{}c@{}}\textbf{\sysp} \\ (Max / Min)\end{tabular}}
\\  \cmidrule(l){4-6} \cmidrule(l){7-9} \cmidrule(l){10-12} \cmidrule(l){13-15} \cmidrule(l){16-18} \cmidrule(l){19-21} 
&   &   & TP & FP & Prec. & TP  & FP  & Prec. & TP & FP  & Prec. & TP & FP  & Prec. & TP & FP  & Prec. & TP & FP  & Prec. \\ \midrule
Small - All          & 181 & 58   & \cellcolortale{84}{49}  & \cellcolorred{32}{23} & \cellcolortale{68}{0.68}  & \cellcolortale{62}{36} & \cellcolorred{25}{12} & \cellcolortale{75}{0.75} & \cellcolortale{37}{22 / 23} & 0 & \cellcolortale{100}{\textbf{1.00}} & \cellcolortale{28}{14} & 0 & \cellcolortale{100}{\textbf{1.00}} & \cellcolortale{97}{57 / 57} & 0 & \cellcolortale{100}{\textbf{1.00}} & \cellcolortale{100}{\textbf{58 / 58}} & 0 & \cellcolortale{100}{\textbf{1.00}} \\
Small - ZKAP & \cellcolor{gray!50}{93} & \cellcolor{gray!50}{20}   & \cellcolor{gray!50}{15}  & \cellcolor{gray!50}{6} & \cellcolor{gray!50}{0.71}  & \cellcolor{gray!50}{15} & \cellcolor{gray!50}{2} & \cellcolor{gray!50}{0.88} & \cellcolor{gray!50}{9 / 9} & \cellcolor{gray!50}{0}\cellcolor{gray!50} & \cellcolor{gray!50}{1.00} & \cellcolor{gray!50}{6} & \cellcolor{gray!50}{0}\cellcolor{gray!50} & \cellcolor{gray!50}{1.00} & \cellcolor{gray!50}{19 / 19} & \cellcolor{gray!50}{0} & \cellcolor{gray!50}{1.00} & \cellcolor{gray!50}{20 / 20} & \cellcolor{gray!50}{0} & \cellcolor{gray!50}{1.00} \\
\hline
Medium - All         & 97 &  7 & \cellcolortale{42}{3}  & \cellcolorred{73}{8} & \cellcolortale{27}{0.27} &  \cellcolortale{14}{1}  & \cellcolorred{94}{15} & \cellcolortale{6}{0.06} & 0 / 0  & 0 & - & 0  & 0 & - & \cellcolortale{100}{\textbf{7 / 7}}  & 0 & \cellcolortale{100}{\textbf{1.00}} & \cellcolortale{100}{\textbf{7 / 7}}  & 0   & \cellcolortale{100}{\textbf{1.00}} \\
Medium - ZKAP         & \cellcolor{gray!50}{42} &  \cellcolor{gray!50}{0} & \cellcolor{gray!50}{0}  & \cellcolor{gray!50}{5} & \cellcolor{gray!50}{0.00} &  \cellcolor{gray!50}{0}  & \cellcolor{gray!50}{8} & \cellcolor{gray!50}{0.09} & \cellcolor{gray!50}{0 / 0}  & \cellcolor{gray!50}{0} & \cellcolor{gray!50}{-} & \cellcolor{gray!50}{0}  & \cellcolor{gray!50}{0} & \cellcolor{gray!50}{-} & \cellcolor{gray!50}{0 / 0}  & \cellcolor{gray!50}{0} & \cellcolor{gray!50}{1.00} & \cellcolor{gray!50}{0 / 0}  & \cellcolor{gray!50}{0}   & \cellcolor{gray!50}{-} \\
\hline
Large - All          & 86 & 17  & \cellcolortale{35}{6}  & \cellcolorred{73}{16} & \cellcolortale{27}{0.27}  & \cellcolortale{19}{4}  & \cellcolorred{57}{4} & \cellcolortale{60}{0.50}  & 0 / 0  & 0 & - & 0 & 0 & - & \cellcolortale{100}{12 / 12}  & 0 & \cellcolortale{100}{\textbf{1.00}}  & \cellcolortale{94}{\textbf{16 / 16}} & 0  & \cellcolortale{100}{\textbf{1.00}}  \\
Large - ZKAP          & \cellcolor{gray!50}{50} & \cellcolor{gray!50}{2}  & \cellcolor{gray!50}{2}  & \cellcolor{gray!50}{ 12} & \cellcolor{gray!50}{ 0.14}  & \cellcolor{gray!50}{0}  & \cellcolor{gray!50}{4} & \cellcolor{gray!50}{0.00}  & \cellcolor{gray!50}{0 / 0}  & \cellcolor{gray!50}{0} & \cellcolor{gray!50}{-} & \cellcolor{gray!50}{0} & \cellcolor{gray!50}{0} & \cellcolor{gray!50}{-} & \cellcolor{gray!50}{0 / 0}  & \cellcolor{gray!50}{0} & \cellcolor{gray!50}{1.00}  & \cellcolor{gray!50}{1 / 1} & \cellcolor{gray!50}{0}  & \cellcolor{gray!50}{1.00}  \\
\hline
Very Large - All     &  88 & 6 & \cellcolortale{83}{\textbf{5}}  & \cellcolorred{84}{27} & \cellcolortale{16}{0.16} &  \cellcolortale{33}{2}  & \cellcolorred{83}{10} & \cellcolortale{17}{0.17}  & 0 / 0  & 0 & - & 0  & 0 & - & \cellcolortale{33}{2 / 2}  & 0 & \cellcolortale{100}{\textbf{1.00}} & \cellcolortale{66}{4 / 3}  & 0 & \cellcolortale{100}{\textbf{1.00}}              \\
Very Large - ZKAP     &  \cellcolor{gray!50}{73} & \cellcolor{gray!50}{5} & \cellcolor{gray!50}{4}  & \cellcolor{gray!50}{15} & \cellcolor{gray!50}{0.21} &  \cellcolor{gray!50}{3}  & \cellcolor{gray!50}{0} & \cellcolor{gray!50}{1.00}  & \cellcolor{gray!50}{0 / 0}  & \cellcolor{gray!50}{0} & \cellcolor{gray!50}{-} & \cellcolor{gray!50}{0}  & \cellcolor{gray!50}{0} & \cellcolor{gray!50}{-} & \cellcolor{gray!50}{1 / 1}  & \cellcolor{gray!50}{0} & \cellcolor{gray!50}{1.00} & \cellcolor{gray!50}{2 / 2}  & \cellcolor{gray!50}{0} & \cellcolor{gray!50}{1.00}              \\
\hline \hline
Total - All          & 452 & 88 & \cellcolortale{75}{63} & \cellcolorred{54}{74} & \cellcolortale{46}{0.46} &  \cellcolortale{53}{43} & \cellcolorred{49}{41}  & \cellcolortale{51}{0.51} & \cellcolortale{28}{22 / 23} & 0 & \cellcolortale{100}{\textbf{1.00}} & \cellcolortale{20}{14} & 0 & \cellcolortale{100}{\textbf{1.00}} & \cellcolortale{86}{79 / 79} & 0  & \cellcolortale{100}{\textbf{1.00}} & \cellcolortale{95}{\textbf{85 / 84}} & 0 & \cellcolortale{100}{\textbf{1.00}} \\
Total - ZKAP          & \cellcolor{gray!50}{258} & \cellcolor{gray!50}{27} & \cellcolor{gray!50}{21} & \cellcolor{gray!50}{38} & \cellcolor{gray!50}{0.35} &  \cellcolor{gray!50}{18} & \cellcolor{gray!50}{14}  & \cellcolor{gray!50}{0.56} & \cellcolor{gray!50}{9 / 9} & \cellcolor{gray!50}{0} & \cellcolor{gray!50}{1.00} & \cellcolor{gray!50}{6} & \cellcolor{gray!50}{0} & \cellcolor{gray!50}{1.00} & \cellcolor{gray!50}{20 / 20} & \cellcolor{gray!50}{0}  & \cellcolor{gray!50}{1.00} & \cellcolor{gray!50}{23 / 23} & \cellcolor{gray!50}{0} & \cellcolor{gray!50}{1.00} 
%\\ 
\\ \bottomrule
\end{tabular}
\end{adjustbox}
\end{table*}

\subsection{RQ1: Effectiveness of \sys} Tab.~\ref{tab:main-result} summarizes the total number of unique bugs detected by each tool across different benchmark categories. We manually check all potential bugs flagged by at least one tool and categorize them into TP (True Positive), which refers to cases where the tool correctly flags ZK circuits as unsafe, and FP (False Positive), which refers to cases where the tool incorrectly flags a test suite as unsafe. Precision is calculated as TP / (TP + FP). The lower gray row in each category corresponds to the results on the ZKAP dataset only~\cite{wen2024practical}. Since there may be undetected bugs not flagged by any tool, we do not report true or false negatives. The \#Bugs column shows the number of unique bugs detected by at least one tool. A bug in a Circom template (\S~\ref{subsec:circom-lang}) may affect all programs that instantiate the template, but we count it as a single bug to avoid double-counting. Moreover, when a vulnerability is detected in multiple categories (for example, if a bug is identified in both a small and a medium circuit because they use the same buggy template), it is counted only in the smallest applicable category to prevent double-counting. Note this way of counting favors existing formal tools due to their limited scalability. We report the maximum and minimum TP of \sys with five different random seeds. Our evaluation shows that \sysp significantly and stably outperforms existing methods, detecting about 96\% (\ndetectedbugs/\nallbugs) of all bugs without any false positives. 

In addition, \sysb, which performs the search without static analysis, still detects more bugs than other existing tools, demonstrating the strength of our core approach. \sysp detects 10\% more bugs overall compared to \sysb, especially in larger circuits.

One clear advantage of \sys over static analyzers like Circomspect and ZKAP is its ability to avoid false positives. This makes \sys more practical as a testing tool, as developers are often inclined to disregard alerts from tools with high false positive rates~\cite{engler2008couple,johnson2013don}.

Our benchmark shows that the performance of ZKAP is worse than that reported in~\cite{wen2024practical}. This stems from not classifying potential logic issues, such as unused inputs, as vulnerabilities. Furthermore, the official implementation of ZKAP is known to crash on certain templates written in recent Circom versions—an issue acknowledged by ZKAP's developers but yet to be fixed. The results of ZKAP with \textit{US} and \textit{USCO} detectors and our modified version of Circomspect with recursive analysis are provided in Appendix~\ref{appendix-additional-experiments}.

Our analysis also reveals that 55.2\% of root causes stem from computation abort bugs, while 44.8\% are due to non-deterministic behavior. This result highlights the importance of comprehensively addressing both types of bugs, whereas Picus and ConsCS focus solely on non-deterministic cases. 

ConsCS detects fewer bugs overall compared to Picus. Although \cite{jiang2025conscs} reports a higher success rate for ConsCS, it is important to note that their evaluation measures both safe and unsafe circuits (demonstrating either security or vulnerability), while we focus solely on vulnerability detection and do not include formal verification of circuit safety.

\noindent
\rqbox{\textbf{Result 1: The effectiveness of \sys}}{Both \sysb and \sysp outperform existing tools, efficiently detecting \ndetectedbugs of \nallbugs vulnerabilities in the benchmark of real-world ZK circuits without false positives, demonstrating its reliable ability to identify critical security flaws in ZK circuits.
}

\subsection{RQ2: Detection Speed}

\begin{figure}
    \centering
    \includegraphics[width=\linewidth]{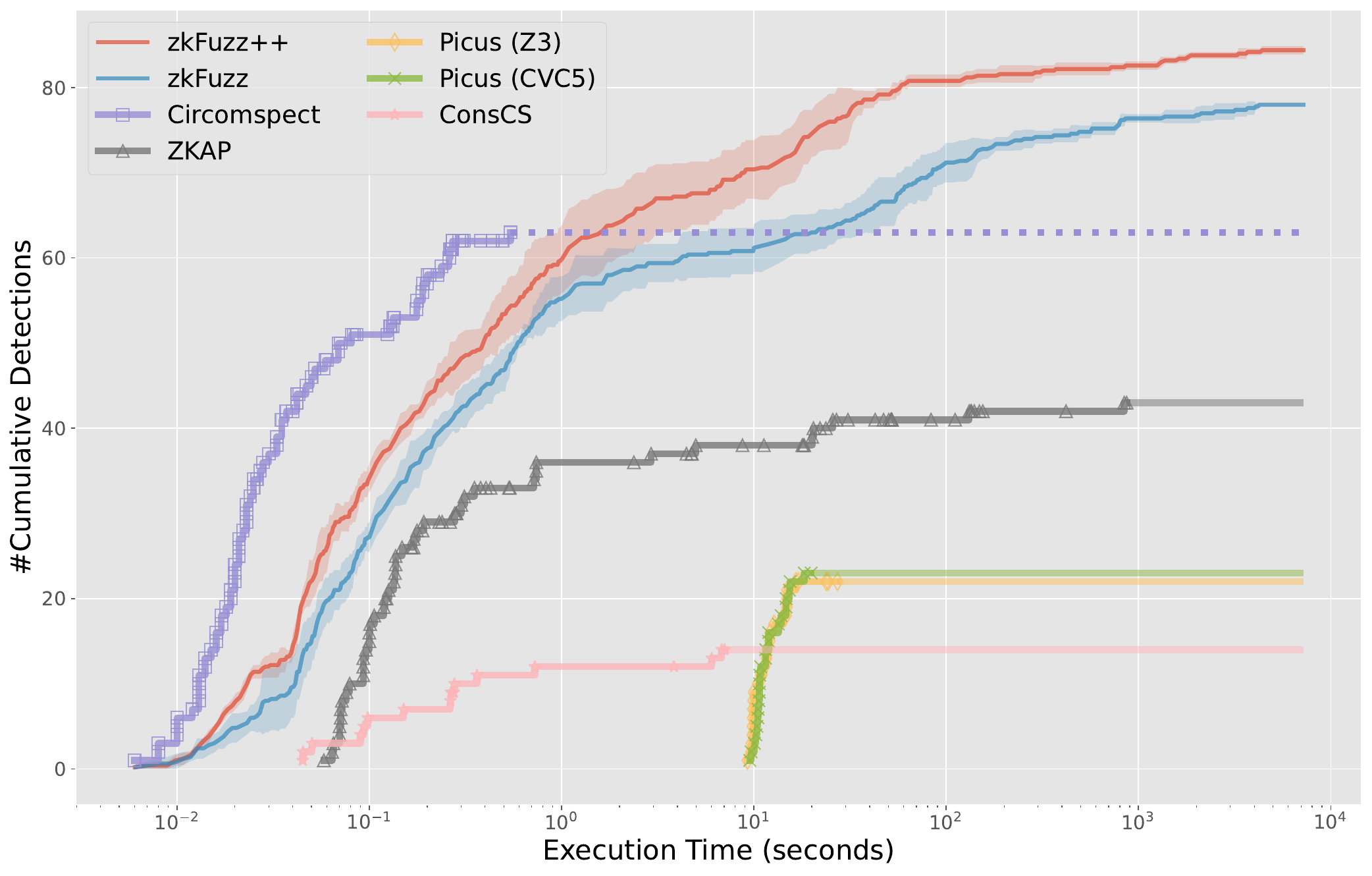}
    \caption{Detection Time Analysis. Both modes of \sys are superior to other methods.}
    \label{fig:execution-time}
\end{figure}

Fig.~\ref{fig:execution-time} shows the cumulative number of unique vulnerabilities detected by each tool over time. For \sysb and \sysp, we depict the average and 1-sigma regions across five trials with different random seeds. Notably, \sysp identifies over 90\% of bugs within the first 100 seconds, although the naive \sysb takes more than one hour to achieve the same number of detected bugs, showing the effectiveness of our target selectors. Although ZKAP is a static analyzer, it suffers from higher performance overhead than \sys due to the overhead of compiling Circom files to LLVM and analyzing its dependency graph. Similarly, Picus exhibits slower performance, relying on expensive queries to SMT solvers. Our results show that ConsCS finds bugs faster than Picus, confirming the findings of \cite{jiang2025conscs}. The dashed line of Circomspect shows that it finishes the analysis of all test cases before the timeout.

\noindent
\rqbox{\textbf{Result 2: Detection Speed of {\sys}}}{
\sysp identifies 90\% of bugs within just 100 seconds, demonstrating its practical effectiveness over traditional formal methods. The comparison with \sys highlights the advantage of target selectors, achieving a speedup of more than 30 times.
}

\subsection{RQ3: Previously Unknown Bugs}
\label{subseq:previously-unknown-bugs}

Out of \ndetectedbugs bugs identified by \sysp, \nnewbugs were previously unknown. Of these, \nconfirmed were confirmed by developers, and \nfixed have already been fixed. One example is shown in Code~\ref{lst:dataencoder}, found in the \textit{passport-zk-circuits}~\cite{passport-zk-circuits}, a Web3 project that raised \$2.5 million. For instance, when \verb|day=4|, the expected assignment is \texttt{\{dayD:0, dayR:4\}} although \texttt{\{dayD:1, dayR:-6\}} also satisfies the constraint.

%\begin{figure}[!ht]
%\centering
%\begin{minipage}[t]{0.54\linewidth}
%\centering
%\lstset{frame=single}
\begin{lstlisting}[language=Circom, style=circomstyle, caption={A previously unknown bug detected by \sys in \textit{passport-zk-circuits}~\cite{passport-zk-circuits}. This circuit is under-constrained.}, label={lst:dataencoder}]
template DateEncoder() {
 signal output encoded;
 signal input  day, month, year;
 signal dayD <-- (day \ 10);
 signal dayR <-- (day % 10);
 dayD * 10 + dayR === day;
 signal dayE <== (dayD * 2**8 + dayR) + 12336;
//Encode month and year in the same manner
 encoded<== yearE * 2**32 + monthE * 2**16 + dayE;
\end{lstlisting}

Another example in Code~\ref{lst:whitelist} detected in the \textit{zk\_whitelist}~\cite{zkwhitelist}. This circuit is expected to abnormally terminate when the two inputs are different by \verb|assert|, while this condition is not included in the constraints. This bug, confirmed by the developer and fixed, does not align with any prior formulations or patterns of ZK circuit vulnerabilities, showing the superiority of our TCCT. The same misuse of \verb|assert| is also found in two circuits of \textit{masa-zkSBT}, confirmed the developer.

\begin{lstlisting}[language=Circom, style=circomstyle, caption={A previously unknown bug found by \sys in \textit{zk\_whitelist}~\cite{zkwhitelist}. Any existing tools cannot capture this bug.}, label={lst:whitelist}]
template Whitelist() {
 signal input addressInDecimal, publicAddress;
 signal output c;
 component ise = IsEqual();
 ise.in[0] <== addressInDecimal;
 ise.in[1] <== publicAddress;
 assert(ise.out==1);
 c <== ise.out; }
\end{lstlisting}

Unlike static analyzers such as ZKAP and Circomspect, \sys can generate concrete counterexamples, significantly reducing the need for manual inspection. The bugs identified by \sys were patched within a median of just 5.0 days after disclosure, underscoring both their critical severity and the prompt response of the maintainers.

\noindent
\rqbox{\textbf{Result 3: New Bugs found by {\sys}}}{
{\sysp} uncovers \nnewbugs previously unknown bugs in real-world circuits, including ones missed by all existing tools due to incomplete formulations.
}

\subsection{RQ4: Ablation Study of Target Selectors}
\label{subsec:ablation}

Tab.~\ref{tab:ablation} shows the results of ablation studies where we assess the impact of removing individual target selectors: \textit{error-based fitness score}, \textit{skewed distribution}, \textit{white list}, \textit{invalid array subscript}, \textit{hash-check} and \textit{zero-division}. We report the average and standard deviation of the cumulative number of unique bugs found at each time checkpoint among five different random seeds. For skewed distribution, we use a less skewed distribution; binary values (0 and 1) with 15\%, small positive integers (2 to 100) with 34\%, larger values near $q$ ($q- 1000$ to $q - 1$) with 50\%, and all other values (11 to $q - 1001$) with 1\%. Fig.~\ref{fig:ablation-detail} also presents the relative execution time for discovering each unique bug, normalized against the default configuration.

Among these factors, the choice of skewed distribution has the largest initial impact. Our default \sysp (Ours) finds 10 more bugs within the first 100 seconds compared to the less skewed distribution. The relative execution time indicates that adopting a less skewed distribution or removing the whitelist slows down detection by over 100×. In contrast, removing checks for invalid array subscripts, hash integrity, or division by zero prevents the fuzzer from detecting several bugs. However, after a reasonable amount of time, the performance of the less skewed distribution also converges to a similar number of unique detected bugs, demonstrating the stability and robustness of \sys regardless of the specific choice of skewed distributions. Additionally, removing static analysis-based target selectors results in several bugs being missed, highlighting their importance in the system's overall effectiveness.

% Please add the following required packages to your document preamble:
% \usepackage{booktabs}
\begin{table*}[!th]
\centering
\caption{Impact of removing individual target selector on \sysp's performance. Each selector contributes significantly to the number of detected bugs and the detection speed, demonstrating their importance.}
\label{tab:ablation}
\begin{tabular}{@{}lcccccc@{}}
\toprule
Time (s) & \begin{tabular}[c]{@{}c@{}}Ours\\ (\sysp)\end{tabular} & Less Skewed Distribution & w.o. White List & w.o. Invalid Array Subscript & w.o. Hash Check & w.o. Zero Division \\ \midrule
10.0 & 71.40 ± 3.44 & 61.60 ± 1.62 & 66.40 ± 3.01 & 69.60 ± 4.08 & 68.40 ± 3.83 & 70.40 ± 3.56 \\
100.0 & 81.80 ± 0.75 & 68.80 ± 0.98 & 77.00 ± 2.19 & 80.00 ± 0.89 & 78.80 ± 0.75 & 80.80 ± 0.75 \\
1000.0 & 83.60 ± 0.49 & 74.80 ± 2.48 & 82.20 ± 0.75 & 82.00 ± 0.63 & 80.60 ± 0.49 & 82.60 ± 0.49 \\
7200.0 & 85.40 ± 0.49 & 84.00 ± 0.89 & 84.60 ± 0.49 & 83.40 ± 0.49 & 82.00 ± 0.00 & 84.60 ± 0.49 \\ \bottomrule
\end{tabular}
\end{table*}

\begin{figure}[!ht]
    \centering
    \includegraphics[width=\linewidth]{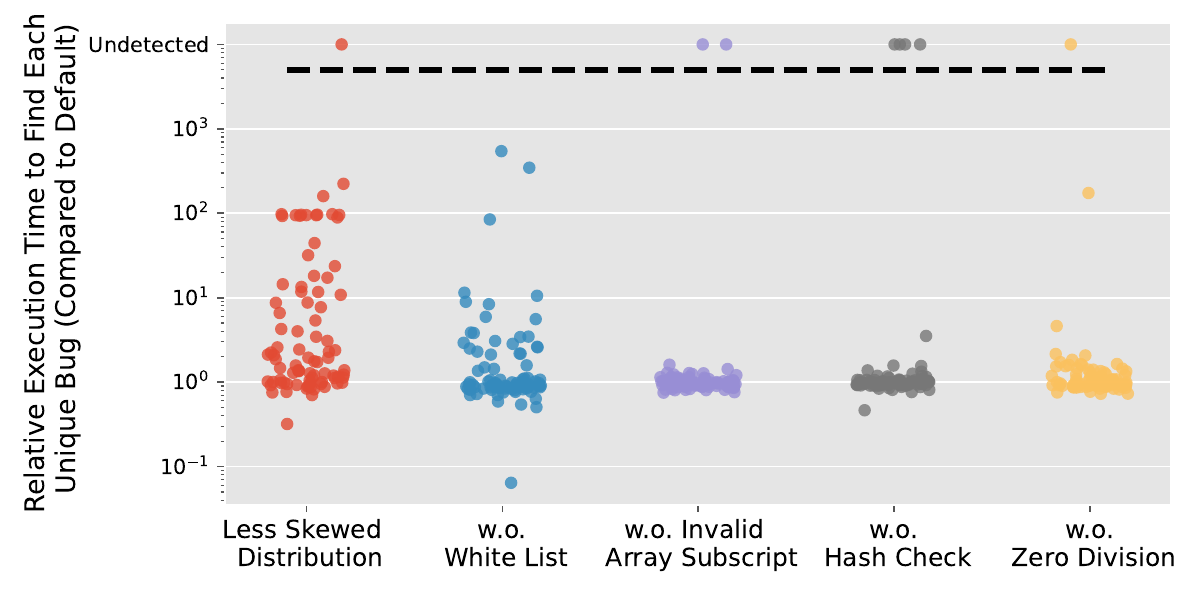}
    \caption{Relative execution time for discovering each unique bug compared to the default setting. Removing each heuristic degrades detection performance by over 100x.}
   \label{fig:ablation-detail}
\end{figure}

\noindent
\rqbox{\textbf{Result 4: Ablation Study of Target Selectors}}{
Each target selector in {\sysp} significantly enhances bug detection by increasing detection speed by more than 100 times and enabling the discovery of bugs that cannot be found without them.
}

\subsection{RQ5: Hyperparameter Sensitivity}
\label{subseq:sensitivity}

\begin{figure*}[!ht]
    \centering
    \includegraphics[width=\linewidth]{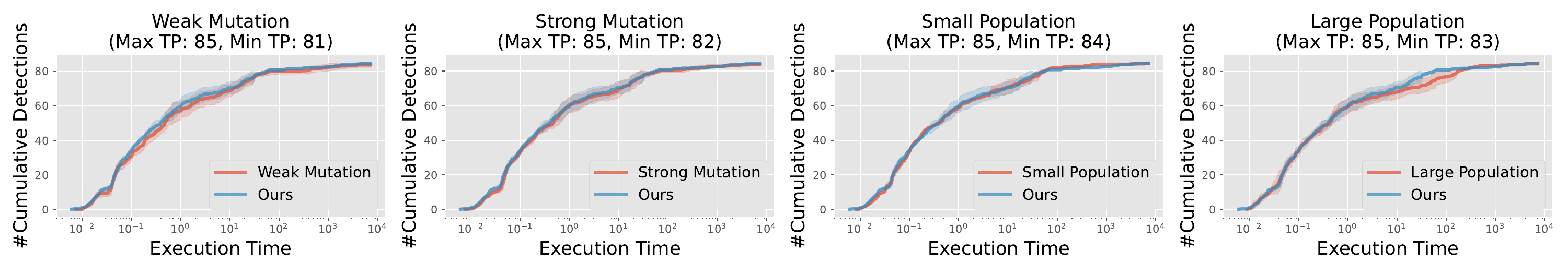}
    \caption{Impact of mutation strength and population size on \sysp. \sysp consistently maintains strong bug detection performance across various hyperparameter settings.}
    \label{fig:mutation-strength}
\end{figure*}

We further examine how mutation strength and population size affect performance. Fig.~\ref{fig:mutation-strength} compares configurations of \sysp using weak mutation (mutation and crossover probabilities of 0.1) versus strong mutation (both set to 0.7) and small population (the number of program mutants and generated inputs of 10) versus large population (both set to 50). Our results show that weak mutation and a large population yield slightly inferior performance compared to the default setting, showing the importance of sufficient exploration in navigating the vast search space since weak mutation slows the speed of exploration while larger populations cause the fuzzer to execute more programs before steering mutants toward promising directions, increasing overhead. Nonetheless, regardless of hyperparameter choices, cumulative detections converge to similar levels.

\noindent
\rqbox{\textbf{Result 5: Stability of \sys}}{
The bug detection performance of our fuzzer remains stable across different hyperparameter choices, identifying the same number of bugs in all settings.
}

\subsection{RQ6: Preliminary Prototype for Noir}
\label{subseq:noir-evaluation}

To demonstrate the language-agnostic capabilities of \sys, we target Noir as a representative modern ZK DSL, which is a next-generation zero-knowledge (ZK) DSL, gaining traction. While Noir can automatically generate constraints for witness computation, developers often write manual constraints to optimize circuit size—enabling more efficient ZK circuits but also introducing opportunities for subtle bugs.

Supporting Noir in \sys presents substantially more difficult challneges than Circom, as Noir’s rich, Rust-inspired syntax—with its advanced type system and user-defined abstractions—demands much more sophisticated parsing and emulation infrastructure, all within an ecosystem that is still in its early stages of development.

Therefore, our preliminary prototype of \sys for Noir leverages its intermediate representations (IRs). Specifically, the Noir compiler lowers the Rust‐like program into two IRs: ACIR, which corresponds to the arithmetic circuit constraints, and Brilling, which corresponds to the unconstrained witness computation logic. As our mutation strategy, we modify the source operand of the \verb|Mov| instructions in the Brilling opcodes to reference special memory regions; our custom memory structure then returns random values with the appropriate data type when those regions are read. In this way, we can mutate the witness‐calculation program without altering the underlying constraints. Although we have not yet implemented any target selectors for this Noir prototype, doing so is clearly feasible with a modest amount of engineering effort—an avenue we leave to future work.

Since there is no existing fuzzers for Noir targeting under or over-constrained circuit or dataset of vulnerable Noir circuits, we evaluate the preliminary implementation of \sys for Noir on 20 representative Noir circuits from 8 open-source libraries, chosen by ourselves based on the number of stars and forks on GitHub and compatibility with the latest version of Noir compiler, where \sys successfully rediscovers a known under-constrained bug~\cite{ecrecover-noir-pr16}.

\noindent
\rqbox{\textbf{Result 6: Language-Agnostic Design of \sys}}{
Our preliminary prototype for Noir (which no other tools support) successfully identified an under-constrained circuit, showing that \sys's potential to generalize beyond a single language.
}

\noindent{\textbf{Ethical Considerations.}} We responsibly disclosed all previously unknown vulnerabilities found in public projects. When possible, we contacted maintainers privately; otherwise, we opened GitHub issues. Most disclosures were made at least 60 days before submission, and the detailed information about any unfixed bugs discovered after that cutoff is not included in this submission.

%% file: subfiles/related_work.tex
%%
%% This is file `sample-sigconf-authordraft.tex',
%% generated with the docstrip utility.
%%
%% The original source files were:
%%
%% samples.dtx  (with options: `all,proceedings,bibtex,authordraft')
%% 
%% IMPORTANT NOTICE:
%% 
%% For the copyright see the source file.
%% 
%% Any modified versions of this file must be renamed
%% with new filenames distinct from sample-sigconf-authordraft.tex.
%% 
%% For distribution of the original source see the terms
%% for copying and modification in the file samples.dtx.
%% 
%% This generated file may be distributed as long as the
%% original source files, as listed above, are part of the
%% same distribution. (The sources need not necessarily be
%% in the same archive or directory.)
%%
%%
%% Commands for TeXCount
%TC:macro \cite [option:text,text]
%TC:macro \citep [option:text,text]
%TC:macro \citet [option:text,text]
%TC:envir table 0 1
%TC:envir table* 0 1
%TC:envir tabular [ignore] word
%TC:envir displaymath 0 word
%TC:envir math 0 word
%TC:envir comment 0 0
%%
%% The first command in your LaTeX source must be the \documentclass
%% command.
%%
%% For submission and review of your manuscript please change the
%% command to \documentclass[manuscript, screen, review]{acmart}.
%%
%% When submitting camera ready or to TAPS, please change the command
%% to \documentclass[sigconf]{acmart} or whichever template is required
%% for your publication.
%%
%%

%\documentclass[../main.tex]{subfiles}

%%
%% end of the preamble, start of the body of the document source.
%\begin{document}

\section{Related Work}
\label{sec:rw}
Zero-knowledge proofs are vital for privacy and trust in decentralized systems, including cryptocurrencies and smart contracts. However, flaws in circuit constraints can compromise security by permitting malicious proofs or rejecting valid ones~\cite{tang2024zero,liang2025sok}. We summarize different approaches to automatically detect various types of ZKP bugs below (see Appendix~\ref{appendix:comparison} for the comparison table).

\noindent{\textbf{Static Analysis.}} These approaches are widely used but suffer from high false positives. For example, Circomspect~\cite{circomspect} flags weak assignments (\verb|<--|) indiscriminately, leading to many false positives. ZKAP~\cite{wen2024practical} improves accuracy using circuit dependence graphs, yet it still identifies unconstrained inputs as potential vulnerabilities, which may be expected behavior. Despite identifying issues, static analyzers can't generate counterexamples, requiring manual verification of flagged bugs.

\noindent{\textbf{Formal Methods.}} Formal methods rigorously verify ZK circuit correctness but struggle with scalability and require manual annotations. Tools like CIVER~\cite{isabel2024scalable} and Constraint Checker~\cite{fan2024snarkprobe} need human input, while others like CODA~\cite{liu2024certifying} require rewriting circuits in another language. Fully automatic tools like Picus~\cite{pailoor2023automated} and ConsCS~\cite{jiang2025conscs} use SMT solvers but still face scalability challenges and offer narrower definitions of under- and over-constrained circuits than ours. For example, they don’t capture under-constrained circuits caused by unexpected inputs~\cite{wen2024practical}.

\noindent{\textbf{Dynamic Testing.}} Fuzzing and mutation testing have been used to detect bugs in ZKP systems~\cite{hochrainer2024fuzzing,leeb2024metamorphic,xiao2025mtzk}, typically focusing on compiler bugs. In contrast, we target security vulnerabilities in individual ZK circuits.

%% file: subfiles/limitation_and_futurework.tex
\section{Limitations and Future Work}

\noindent{\textbf{Limitations.}} \sys adopts a fuzzing approach for bug detection and thus inherits certain limitations. For example, the execution time is proportional to the size of the target program, and slicing and partially executing the target might improve \sys. Our prototype for Noir also remains preliminary, lacking guiding and target selection methods.

\noindent{\textbf{Future Work.}} A primary avenue for improvement is to enhance Noir support by adding more guiding and target selectors. {In addition, since TCCT does not depend on any particular DSL and many DSLs compile down to standardized constraint formats such as R1CS, TCCT - and consequently zkFuzz - can be applied across a wide range of ZK languages by targeting these common trace/constraint representations, regardless of surface-level differences in syntax or semantics.} Integrating \sys with formal verification, static analyzers, and ML-based bug detection~\cite{zhang2025low,she2019neuzz,he2024large,gan2024defialigner} might also further boost its scalability and flexibility.

%% file: subfiles/conclusion.tex
\section{Conclusion}

We present \sys, a fuzzing framework for finding ZK circuit bugs based on \model—a language-independent bug model. \sys uses guided mutations and inputs to detect under- and over-constrained circuits without false positives, uncovering \nnewbugs previously unknown bugs in real-world circuits.

%\ifSubfilesClassLoaded{%
%\bibliographystyle{ACM-Reference-Format}
%\bibliography{ref}
%}{
%}

%\end{document}

%% file: subfiles/appendix.tex
\appendices
\section{Proofs}
\label{appendix:proof}

\begin{proof}[Proof of Theorem~\ref{thm:under-constrained-detection-np-complete}]

First, we will reduce the Boolean Satisfiability Problem (SAT) to the problem of deciding whether a program is under-constrained (or over-constrained). SAT is known to be NP-complete.

(\textbf{Detection of Under-Constrained Instance}): Let $\phi = \bigwedge^{m}_{j=1} \phi_j$ be a Boolean formula in conjunctive normal form (CNF) with variables $x_1, \ldots, x_n$, where $\phi_j$ is the $j$-th conjunction of $\phi$ and $m$ is the total number of conjunctions. Then, we construct a program $\langle \mathcal{P}, \mathcal{C} \rangle$ as follows:

\begin{align*}
 &\mathcal{P}: \mathbb{F}^{n}_2 \rightarrow ((\mathbb{F}^{0}_2 \cup \{\bot\}) \times (\mathbb{F}_2^{0} \cup \{\bot\})) \\
 &\mathcal{C}: \{\mathbb{F}^{n}_2 \times \mathbb{F}^{0}_2 \times \mathbb{F}^{0}_2 \rightarrow \{0, 1\}\} \\
 \cline{1-2}
 &\forall{(x_1, \ldots, x_n) \in \mathbb{F}^n_2}, \quad \mathcal{P}((x_1, \ldots, x_n)) = (\bot, \bot) \\
&\mathcal{C} = \phi
\end{align*} Since $\mathcal{P}$ is a program that always returns ($\bot$, $\bot$), $\mathcal{T}(\mathcal{P}) = \emptyset$.

Now, we claim that $\langle \mathcal{P}, \mathcal{C} \rangle$ is under-constrained if and only if $\phi$ is satisfiable.

($\Rightarrow$) If $\langle \mathcal{P}, \mathcal{C} \rangle$ is under-constrained, there exists $(x, y) \in \Pi(\mathcal{S}(\mathcal{C}))$ such that $(x, y) \notin \Pi(\mathcal{T}(\mathcal{P}))$. This means that there exists at least one $(x, (), y) \in \mathcal{S}(\mathcal{C})$. This implies that $x$ satisfies $C = \phi$. Thus, $x$ is a satisfying assignment for $\phi$.

($\Leftarrow$) If $\phi$ is satisfiable, let $x'$ be a satisfying assignment. Then, we have $(x', ()) \in \Pi(\mathcal{S}(\mathcal{C}))$ because $x'$ satisfies all clause constraints. However, $(x', ()) \notin \Pi(\mathcal{T}(\mathcal{P}))$ because $\mathcal{T}(\mathcal{P})$ is empty. Thus, $\langle \mathcal{P}, \mathcal{C} \rangle$ is under-constrained.

Since the entire reduction can be performed in polynomial time, we have shown that deciding whether $\langle \mathcal{P}, \mathcal{C} \rangle$ is under-constrained is at least as hard as SAT. Therefore, deciding whether a program is under-constrained is NP-hard.

$\quad$

(\textbf{Detection of Over-Constrained Instance}): Likewise, we construct a program $\langle \mathcal{P}, \mathcal{C} \rangle$ as follows:

\begin{align*}
 &\mathcal{P}: \mathbb{F}^{n}_2 \rightarrow (\mathbb{F}^{0}_2 \cup \{\bot\}) \times (\mathbb{F}_2^{1} \cup \{\bot\}) \\
 &\mathcal{C}: \{\mathbb{F}^{n}_2 \times \mathbb{F}^{0}_2 \times \mathbb{F}^{1}_2 \rightarrow \{0, 1\}\} \\
 \cline{1-2}
 &\mathcal{P}(x) = \begin{cases}
        ((), (1)) &\quad \text{if all clauses in $\phi$ is satisfied with $x$} \\
     (\bot, \bot) &\quad \text{otherwise}
 \end{cases}  \\
&\mathcal{C} = false
\end{align*} \noindent Since $\mathcal{C}$ always returns false, $\mathcal{S}(\mathcal{C}) = \emptyset$. 

Now, we claim that $\langle \mathcal{P}, \mathcal{C} \rangle$ is over-constrained if and only if $\phi$ is satisfiable.

($\Rightarrow$) If $\langle \mathcal{P}, \mathcal{C} \rangle$ is over-constrained, there exists $(x, z, y) \in \mathcal{T}(\mathcal{P})$ such that $(x, z, y) \notin \mathcal{S}(\mathcal{C})$. This implies that $x$ satisfies all clauses in $\phi$, meaning $x$ is a satisfying assignment for $\phi$.

($\Leftarrow$) Suppose $\phi$ is satisfiable, and let $x'$ be a satisfying assignment. Then, we have $(x', (), (1)) \in \mathcal{T}(\mathcal{P})$ because $x'$ satisfies all clauses in $\phi$. However, $(x', (), (1)) \notin \mathcal{S}(\mathcal{C})$ because $\mathcal{S}(\mathcal{C})$ is empty. Thus, $\langle \mathcal{P}, \mathcal{C} \rangle$ is over-constrained.

Since the entire reduction can be performed in polynomial time, we have shown that deciding whether $\langle \mathcal{P}, \mathcal{C} \rangle$ is over-constrained is at least as hard as SAT. Therefore, deciding whether a program is over-constrained is NP-hard.

(\textbf{Complexity of \model}): Next, we prove TCCTs's co-NP-completeness by reducing the Boolean tautology problem, a known co-NP-complete problem, to it.

Let $\phi_1$ and $\phi_2$ be Boolean formulas with variables $x = (x_1, \ldots, x_n)$. The Boolean equivalence problem asks whether  $\phi_1 \equiv \phi_2$, i.e., $\phi_1$ and $\phi_2$ evaluate to the same value for all possible assignments. The Boolean tautology problem aims to determine whether all possible assignments to a Boolean formula yield true. 

\begin{lemma}
    The Boolean equivalence problem is reducible to the Boolean tautology problem.
\end{lemma}

\begin{proof}
   $\phi_1 \equiv \phi_2$ if and only if $(\phi_1 \land \phi_2) \lor ((\neg \phi_1) \land (\neg \phi_2))$ is a tautology.
\end{proof}

\begin{lemma}
    The Boolean tautology problem is reducible to the Boolean equivalence problem.
\end{lemma}

\begin{proof}
    $\phi_1$ is a tautology if and only if $\phi_1 \equiv \hbox{TRUE}$
\end{proof}

By these lemmas, the Boolean tautology problem and the Boolean equivalence problem are polynomial-time reducible to each other and, thus, are of equivalent complexity.

Given two Boolean formulas $\phi_1$ and $\phi_2$, we construct an instance of TCCT $\langle \mathcal{P}, \mathcal{C} \rangle$ over $\mathbb{F}_2$ as follows:

\begin{align*}
 &\mathcal{P}: \mathbb{F}^{n}_2 \rightarrow (\mathbb{F}^{0}_2 \cup \{\bot\}) \times (\mathbb{F}_2^{1} \cup \{\bot\}) \\
 &\mathcal{C}: \{\mathbb{F}^{n}_2 \times \mathbb{F}^{0}_2 \times \mathbb{F}^{1}_2 \rightarrow \{0, 1\}\} \\
 \cline{1-2}
 &\mathcal{P}((x_1, \ldots, x_n)) = \begin{cases}
     ((), (1)) &\text{if $(x_1, \ldots, x_n)$ satisfies $\phi_1$} \\
     ((), (0)) &\text{otherwise} \\
 \end{cases} \\
&\mathcal{C}(x, z, y) = \begin{cases}
            1 &\text{if $x$ satisfies $\phi_2$ and $y=1$} \\
            1 &\text{if $x$ does not satisfy $\phi_2$ and $y=0$} \\
            0 &\text{otherwise}
        \end{cases}
\end{align*}

\noindent Since there are no intermediate values, Eq.~\ref{tcct-condition} is equivalent to $\mathcal{T}(\mathcal{P}) = \mathcal{S}(\mathcal{C})$. By the definition of $\mathcal{P}$ and $\mathcal{C}$, we have that $\phi_1 \equiv \phi_2$ if and only if $\langle \mathcal{P}, \mathcal{C} \rangle$ is well-constrained.

This reduction can be performed in polynomial time. Since the Boolean tautology problem is co-NP-complete~\cite{arora2009computational} and we have shown a polynomial-time reduction to TCCT, we conclude that it is co-NP-complete.

\end{proof}

\section{Additional Examples of ZK Program}
\label{appendix:examples}

\indent{\textbf{Correct Implementation of RShift1.}} Code~\ref{lst:correct-1brshift} presents a secure implementation of a 1-bit right shift in Circom. The input is first converted into a bit array, the shift operation is applied to this array, and the result is then converted back into the output.

\begin{lstlisting}[language=Circom, style=circomstyle, caption={Secure implementation of 1-bit right shift.}, label={lst:correct-1brshift}]
template RShift1(N) {
 signal input x;
 signal output y;
 signal out <== Num2Bits(N)(x);
 signal bits[N-1];
 for(var i=0; i<N-1; i++) {bits[i] <== out[i+1];}
 y <=== Bits2Num(N - 1)(bits); }
\end{lstlisting}

\noindent{\textbf{LessThan.}} While explicit constraints are readily apparent in circuit designs, implicit assumptions used in Circom's de facto standard library, \textit{circomlib}~\cite{circomlib}, can also introduce subtle under-constrained bugs if not properly accounted for. For example, \texttt{LessThan} circuit from \textit{circomlib}, which checks whether the first input is lower than the second input, implicitly assumes that both inputs are represented by n bits or fewer. The critical vulnerability arises from the potential overflow in Line 7 of \texttt{LessThan} when input a is excessively large. This is due to the modular arithmetic operations over $\bmod q$. To illustrate, consider a case where $q = 17$, $a = 10$, and $b = 1$. The computation in Line 7 of \texttt{LessThan} becomes $10 + (1 \ll 3) - 1 \equiv 0 \ \bmod 17$, erroneously indicating that $a < b$. 

This class of vulnerabilities, stemming from the violation of implicit assumptions of templates, can be incorporated into our TCCT framework. The easiest way for this case is to insert one additional assert to the computation $\mathcal{P}$: \texttt{assert((in[0] < in[1] and out == 1) or (in[0] >= in[1] and out == 0))}.

Consider the real-world ZK program in Code~\ref{lst:dsk} from \textit{zk-regex}~\cite{zkregex} which validates the range of the given characters and is expected to abnormally terminate when they contain values larger than 255: This template uses the \texttt{LessThan} circuit, although the \texttt{EmailAddrRegex} template does not validate the bitwidth of the inputs provided to \texttt{LessThan}. To illustrate the attack vector, suppose $q = 401$ and \texttt{msg[0]} $=400$. The computation in Line 7 of \texttt{LessThan} becomes $400 + (1 \ll 8) - 255 \equiv 0 \ \bmod 401$. Hence, \texttt{LessThan} outputs $1$, erroneously passing the range check of \texttt{EmailAddrRegex}. Our fuzzer successfully discovered this bug, which the developer subsequently confirmed.

\begin{figure}[!ht]
\begin{minipage}[t]{0.415\linewidth}
\begin{lstlisting}[language=Circom, style=circomstyle,  caption={\texttt{LessThan} template from \textit{circomlib}~\cite{circomlib}. This template assumes that the inputs are at most \texttt{n} bits wide. If the inputs exceed \texttt{n} bits, the comparison result may be incorrect.}, label={lst:lessthan}]
template LessThan(n){
 assert(n <= 252);
 signal input in[2];
 signal output out;
 component n2b = Num2Bits(n+1);
 n2b.in <== in[0] + (1<<n) - in[1];
 out <== 1 - n2b.out[n];}
\end{lstlisting}
    \end{minipage}
\hfill
\begin{minipage}[t]{0.534\linewidth}
\begin{lstlisting}[language=Circom, style=circomstyle,  caption={Example of an under-constrained circuit from \textit{zk-regex}~\cite{zkregex}. \texttt{LessThan} implicitly assumes that the bit length of both inputs are not longer than \texttt{n}, although \texttt{LessThan} itself does not constrain the range of inputs.}, label={lst:dsk}]
template EmailAddrRegex(n){
 signal input msg[n];

 signal in_range_checks[n];
 for (var i=0; i<n; i++) {
  in_range_checks[i] <== LessThan(8)([msg[i], 255]);
  in_range_checks[i] === 1;
  // the rest is omitted
\end{lstlisting}
    \end{minipage}
\end{figure}

\noindent{\textbf{Over-Constrained Circuit.}} As described in \S~\ref{subsec:circom-lang}, the default configuration of the Circom compiler does not provide an operator that adds a condition to the constraints $\mathcal{C}$ without adding the corresponding assignment or assertion to the computation $\mathcal{P}$. Consequently, the trace set $\mathcal{T}(\mathcal{P})$ is always a subset of the constraint satisfaction set $\mathcal{S}(\mathcal{C})$, meaning that over-constrained circuits cannot exist.

%\iffalse
\begin{table}[!ht]
    \begin{minipage}[t]{0.47\linewidth}
    \begin{lstlisting}[language=Circom, style=circomstyle,  caption={Example of over-constrained circuit. Suppose \texttt{assert} is not added in the computation for \texttt{===}.}, label={lst:over-div}]
template SplitReward(){
 signal input x;    
 signal z;
 signal output y;  
 z <-- x \ 2;
 z * 2 === x;
 y <== z + 1;}
    \end{lstlisting}
    \end{minipage}
    \hfill
        \centering
    \begin{minipage}[t]{0.48\linewidth}
        \centering
        \caption{Trace and Constraint Satisfaction Sets of Code~\ref{lst:over-div} ($q=5$).}
        \label{tab:over-div}
        \begin{tabular}{cccccc}
        \hline
        \multicolumn{3}{c}{$\mathcal{T}(\mathcal{P})$}         & \multicolumn{3}{c}{$\mathcal{S}(\mathcal{C})$} \\ \hline
        x & z & \multicolumn{1}{c|}{y} & x    &  z  & y       \\ \hline
        0 & 0 & \multicolumn{1}{c|}{1}   & 0      &  0 & 1         \\
        \textcolor{red}{1} & \textcolor{red}{0} & \multicolumn{1}{c|}{\textcolor{red}{1}}   & \textcolor{red}{1}    &     \textcolor{red}{3} &  \textcolor{red}{4}        \\
        2 & 1 & \multicolumn{1}{c|}{2}   & 2     &  1  & 2         \\
        \textcolor{red}{3} & {\textcolor{red}{1}} & \multicolumn{1}{c|}{\textcolor{red}{2}}   & \textcolor{red}{3}    &  \textcolor{red}{4}   & \textcolor{red}{0}         \\
        4 & 2 & \multicolumn{1}{c|}{3}   & 4   &  2    & 3         \\ \bottomrule
        \end{tabular}
    \end{minipage}
\end{table}

%\junfeng{this seems quite perculiar.}
However, enabling the \texttt{constraint} \texttt{\_assert} \texttt{\_disabled} flag can lead to over-constrained circuits. Code~\ref{lst:over-div} presents a modified version of the \texttt{Reward} from~\cite{wen2024practical}, showing a ZK program that performs integer division. Note that "$\backslash$" operator in Circom is the integer division, while "$*$" operator is the multiplication modulo $q$. Suppose $q=5$ and \texttt{constraint\_assert\_disabled} is turned on. Then, Tab.~\ref{tab:over-div} shows the trace and constraint satisfaction sets. $\mathcal{S}(\mathcal{C})$ does not contain $(\text{x}=1, \text{z}=0, \text{y}=1)$ and $(\text{x}=3, \text{z}=1, \text{y}=2)$ of $\mathcal{T}(\mathcal{P})$, meaning that it is over-constrained ($\mathcal{T}(\mathcal{P}) \setminus \mathcal{S}(\mathcal{C}) \neq \emptyset$). Note that it is also under-constrained.

\noindent{\textbf{Hash-Check Pattern.}} Code~\ref{lst:hash-check-real} illustrates the real-world circuit implementing the hash-check pattern from the \textit{maci} project~\cite{maci2020}. Although this circuit is under-constrained due to insecure call of \texttt{LessThan} within \texttt{Quin} \texttt{Tree}\texttt{InclusionProof}, detecting this issue without the hash-check is difficult, as finding a valid input and Merkle root pair requires inverting the hash function.

\begin{lstlisting}[language=Circom, style=circomstyle, caption={Real example of Hash Check pattern from \textit{maci}~\cite{maci2020}. Determining a valid pair of input and Merkle root requires inverting the hash function.}, label={lst:hash-check-real}]
template QuinLeafExists(levels){
 var LEAVES_PER_NODE = 5;
 var LEAVES_PER_PATH_LEVEL = LEAVES_PER_NODE - 1;
 var i, j;
 
 signal input leaf, root;
 signal input path_index[levels];
 signal input path_elements[levels][LEAVES_PER_PATH_LEVEL];
 
 // Verify the Merkle path
 component verifier = QuinTreeInclusionProof(levels);
 verifier.leaf <== leaf;
 for (i = 0; i < levels; i ++) {
    verifier.path_index[i] <== path_index[i];
    for (j = 0; j < LEAVES_PER_PATH_LEVEL; j ++) {
        verifier.path_elements[i][j] <== path_elements[i][j];}}
 root === verifier.root;
}
\end{lstlisting}

\section{Non-Deterministic Witness Generation}
\label{appendix:non-deterministic-witness}

{TCCT’s generality allows it to avoid false positives that arise in prior definitions like Picus~\cite{pailoor2023automated}, Ac4~\cite{chen2024ac4}, and Conscs~\cite{jiang2025conscs} when dealing with inherently non-deterministic programs. While currently widely used DSLs like Circom are deterministic, it is plausible that future DSLs like CODA [21] might introduce inherent non-determinism (e.g., through concurrency). In such a setting, a program $\mathcal{P}$ could legitimately produce multiple valid outputs \texttt{y}, with constraints $\mathcal{C}$ designed to accept exactly those outputs. In this case, TCCT correctly classifies $(\mathcal{P}, \mathcal{C})$ as well-constrained, while prior works will produce a false positive. To understand the effect, consider the following Circom-like program inspired by the real Decoder circuit in circomlib. Here, \texttt{random\_binomial} non-deterministically returns either 0 or 1. The Decoder takes input \texttt{inp} and produces two outputs:}

\begin{itemize}
    \item {\texttt{out}, a binary array decoded from \texttt{inp}, and}
    \item {\texttt{success}, a boolean flag indicating whether decoding succeeds (specifically, decoding is successful if and only if the \texttt{inp}-th entry of out is 1).}
\end{itemize}

{In this case, although parts of the output (\texttt{out}) are non-deterministic, the constraints ensure validity only when the values remain consistent with the program’s definition. Under TCCT, such a program is therefore not classified as under-constrained. In contrast, prior formulations that consider only the constraints—without accounting for the program semantics—would incorrectly label this circuit as under-constrained, leading to a false positive.}

\begin{lstlisting}[language=Circom, style=circomstyle, caption={{Example of non-deterministic but well-constrained ZK program. TCCT correctly concludes that this is well-constrained.}}, label={lst:non-det}]
template Decoder(w) {
    signal input inp;
    signal output out[w], success;
    Var lc = 0;
    
    for (var i = 0; i < w; i++) {
	out[i] <-- random_binomial();
       out[i] * (inp - i) === 0;
       lc = lc + out[i];
    }
    lc ==> success;
    success * (success - 1) === 0;
}
\end{lstlisting}

\section{Additional Experiments and Analysis}
\label{appendix-additional-experiments}

\paragraph{RQ7: Alternative Mutation Strategies} 
\label{appendix:optionalmutation}

In addition, we explore multiple strategies for input generation and program mutation. First, we apply a genetic algorithm to optimize input generation, using an error-based fitness score similar to the optimization in program mutation. Specifically, the fitness score of an input is defined as the minimum score across multiple program mutants. As shown in the left figure of Fig.~\ref{fig:mutation-strategy}, although performance eventually converges to the same level as the default setting, it remains worse throughout the process. This suggests that, due to the vast search space, guiding only program mutation with the fitness function, while feeding generated inputs, achieves a better balance between exploration and exploitation than guiding both program mutation and input generation with the fitness function, which skews more toward exploitation.

\begin{figure}[!ht]
    \centering
    \includegraphics[width=\linewidth]{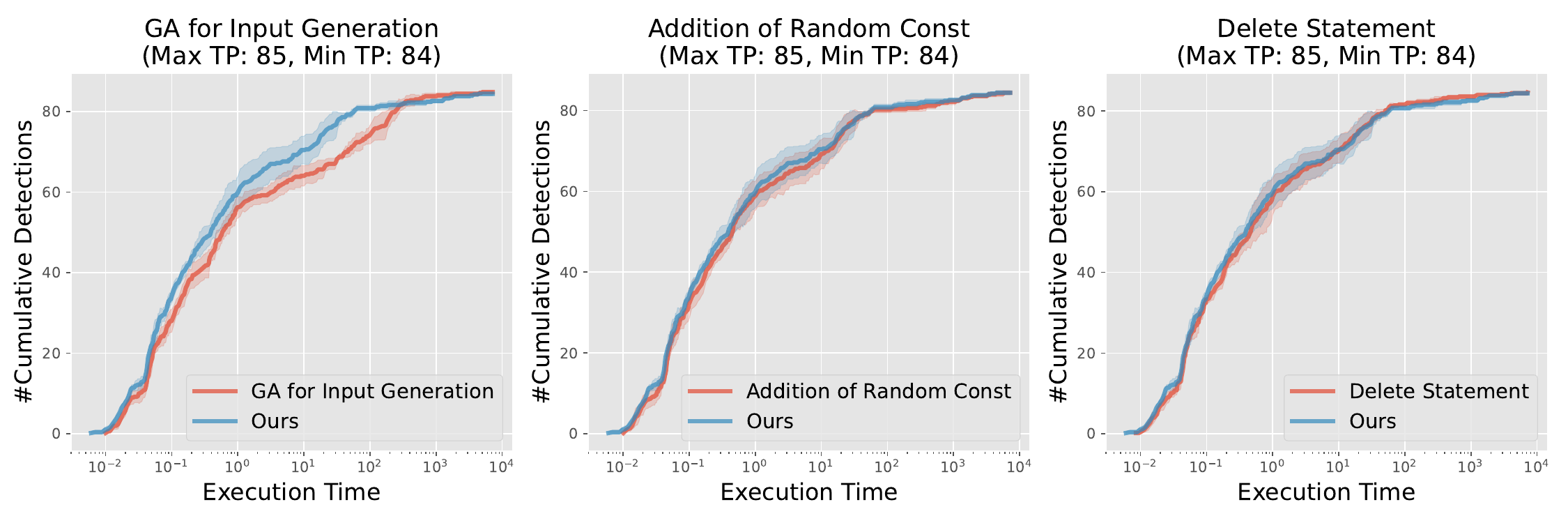}
    \caption{Various input generation and program mutation strategies.}
    \label{fig:mutation-strategy}
\end{figure}

We also investigate two additional program mutation strategies: (1) adding a random constant to the right-hand side of a weak assignment, such as transforming \texttt{z <-- x} into \texttt{z <-- x + 1234}, and (2) removing the statement of the weak assignment entirely. In Circom, all signals are initialized to zero upon declaration, so eliminating a weak assignment ensures that the left-hand variable remains zero throughout execution. These new strategies are applied with a probability of 0.2. However, their performance remains nearly identical to the default setting.

\subsection{RQ8: Performance of Other Baselines} 
\label{appendix:naive-performance}

Tab.~\ref{tab:modified-tools} presents the performance of the modified Circomspect, which recursively analyzes all internally called templates, and ZKAP with the Unconstrained-Signal detector. While recursive Circomspect detects all vulnerabilities, it also increases false positives, causing precision to drop sharply to 0.33. As discussed in \S~\ref{sec:formulation}, Circomspect cannot detect misuse of assert. However, in this dataset, it happens to identify all circuits containing real bugs, mainly because it incorrectly flags code that is not actually problematic. For ZKAP, the Unconstrained-Signal detector does not contribute to finding new bugs in our benchmarks.

% Please add the following required packages to your document preamble:
% \usepackage{booktabs}
\begin{table}[!ht]
\centering
\caption{Performance of Circomspect with recursive analysis and ZKAP with the US and the USCO detector.}
\label{tab:modified-tools}
\begin{tabular}{@{}lcccccc@{}}
\toprule
Constraint Size & \multicolumn{3}{c}{\begin{tabular}[c]{@{}c@{}}Circomspect\\ (Recursive)\end{tabular}} & \multicolumn{3}{c}{\begin{tabular}[c]{@{}c@{}}ZKAP\\ (with US and USCO)\end{tabular}} \\ \cmidrule(l){2-4}  \cmidrule(l){5-7}
 & TP & FP & Prec. & TP & FP & Prec. \\ \midrule
Small & 58 & 28 & 0.67 & 37 & 23 & 0.61 \\
Medium & 7 & 23 & 0.23 & 2 & 29 & 0.06 \\
Large & 17 & 30 & 0.36 & 4 & 8 & 0.33 \\
Very Large & 6 & 100 & 0.05 & 2 & 29 & 0.06 \\ \midrule
Total & 88 & 181 & 0.33 & 45 & 89 & 0.34 \\ \bottomrule
\end{tabular}
\end{table}

\subsection{RQ9: Impact Assessment of Vulnerabilities} 

Table~\ref{tab:bug-impact-assessment} presents an analysis of the impacts of previously unknown bugs discovered by \sys, categorizing each unique bug into seven types: (1) \textit{incorrect primitive operations}, such as faulty integer division or string matching; (2) \textit{game cheating}, such as falsifying the winner; (3) \textit{identity forgery}, such as bypassing age verification; (4) \textit{cryptographic computation bugs}, such as incorrect hash or ciphertext computations; (5) \textit{ownership forgery}, such as maliciously claiming possession of specific data; (6) \textit{incorrect reward calculation}, such as errors in incentive mechanisms; and (7) others. The most common type is ownership forgery, mainly due to incorrect range checks and misuse of the \verb|assert| statement.

\begin{table}[]
\centering
\caption{Bug impact categorization in previously unknown ZK bugs identified by \sys.}
\label{tab:bug-impact-assessment}
\begin{tabular}{@{}lc@{}}
\toprule
Type & \#Bugs \\ \midrule
Incorrect Primitive Operation & 11 \\
Game Cheating & 7 \\
Identity Forgery & 8 \\
Cryptographic Computation Bug & 4 \\
Ownership Forgery & 15 \\
Incorrect Reward Calculation & 7 \\
Others & 7 \\ \bottomrule
\end{tabular}
\end{table}

\subsection{RQ10: Over-Constrained Circuits} 
\label{appendix:oc}

To validate \sys's ability to detect over-constrained circuits, we conduct experiments with the \texttt{constraint\_} \texttt{assert\_} \texttt{disabled} flag enabled, allowing over-constrained circuits to exist. In this mode, the equality constraint \verb|===| adds a condition to the constraints set without inserting an \verb|assert| into the program. Consequently, developers likely need to add the corresponding \verb|assert| manually in the Circom file to prevent over-constrained circuits. Since most circuits are designed to function under the Circom compiler’s default configuration, enabling this flag renders many circuits over-constrained. In our benchmark, \sys identifies 258 such cases out of \nallcircuits test cases. However, because these circuits typically run under the default compiler settings, such issues are unlikely to pose a security risk in practice.

\section{Comparison with Existing Works}
\label{appendix:comparison}

{Tab.~\ref{tab:comparison} presents a comparison between the capabilities of existing works and our \sys along different dimensions, demonstrating its superiority.}

\begin{table}[!th]
\centering
\caption{Comparison between \sys and existing detection methods. A partial circle represents partial capability. \sys demonstrates superior performance across all criteria, offering comprehensive automatic testing capabilities.}
\label{tab:comparison}
\begin{tabular}{@{}lccccc@{}}
\toprule
Method      & \begin{tabular}[c]{@{}l@{}}Auto-\\matic\end{tabular} & \begin{tabular}[c]{@{}l@{}}Under-\\Const-\\rained\end{tabular} & \begin{tabular}[c]{@{}l@{}}Over-\\Const-\\rained\end{tabular} & \begin{tabular}[c]{@{}l@{}}Counter\\ Example\end{tabular} & \begin{tabular}[c]{@{}l@{}}False\\ Posit-\\ve\end{tabular} \\ \midrule
CIVER~\cite{isabel2024scalable} & \xmark & \cmark  & \cmark & \cmark  & \safe \\
CODA~\cite{liu2024certifying} & \xmark & \cmark  & \cmark & \cmark  & \safe \\
\begin{tabular}[l]{@{}l@{}}Constraint\\- Checker\end{tabular}~\cite{fan2024snarkprobe} & \xmark & \cmark  & \cmark & \cmark  & \safe \\
Circomspect~\cite{circomspect} & \cmark & \partialcircle        &   \partialcircle     & \xmark  & \dangerous \\
Picus~\cite{pailoor2023automated}  & \cmark &\partialcircle  & \xmark & \cmark  & \safe \\
ConsCS~\cite{jiang2025conscs}  & \cmark &\partialcircle  & \xmark & \cmark  & \safe \\
AC4~\cite{chen2024ac4}  & \cmark &\partialcircle  & \partialcircle & \cmark  & \safe \\
ZKAP~\cite{wen2024practical} & \cmark &\partialcircle  & \partialcircle & \xmark  & \dangerous \\
\textbf{\sys (Ours)} & \cmark & \cmark  & \cmark & \cmark  & \safe \\ \bottomrule
\end{tabular}
\end{table}